\documentclass[11pt]{article}  
\usepackage{amsmath}
\usepackage{amssymb}
\usepackage{amsfonts}
\usepackage{amsthm}
\usepackage{graphicx}
\usepackage{color}
\usepackage{clrscode}
\usepackage{subfig}
\usepackage{hyperref}
\usepackage{cite}
\usepackage{comment}
\usepackage{transparent}
\usepackage[nofillcomment,noend]{algorithm2e}
\usepackage[top=1.4in, bottom=1.4in, left=1.7in, right=1.7in]{geometry}

\newtheorem{theorem}{Theorem}[section]
\newtheorem*{nunumtheorem}{Theorem}
\newtheorem{lemma}[theorem]{Lemma}
\newtheorem{corollary}[theorem]{Corollary}

\newtheorem{definition}[theorem]{Definition}
\newtheorem{conjecture}[theorem]{Conjecture}

\pdfpageattr {/Group << /S /Transparency /I true /CS /DeviceRGB>>}
\newcommand{\fullgridgraph}{G^\mathrm{f}}\newcommand{\bindinggraph}{G^\mathrm{b}}

\title{Intrinsic universality in tile self-assembly \\ requires cooperation}

\author{
  Pierre-Etienne Meunier%
    \thanks{LAMA, Universit\'e de Savoie. \protect\url{pierreetienne.meunier@univ-savoie.fr}.}
\and
  Matthew J. Patitz%
    \thanks{Department of Computer Science and Computer Engineering, University of Arkansas,
      \protect\url{patitz@uark.edu} 
      Supported in part by National Science Foundation Grant CCF-1117672.}
\and
  Scott M. Summers%
    \thanks{Department
of Computer Science and Software Engineering, University of Wisconsin--Platteville, Platteville, WI 53818, USA.
\protect\url{summerss@uwplatt.edu}.}
\and
  Guillaume Theyssier%
\thanks{LAMA, CNRS and Universit\'e de Savoie, \protect\url{guillaume.theyssier@univ-savoie.fr} 
Supported in part by grant 'Agence Nationale de la Recherche ANR-09-BLAN-0164' }
\and
  Andrew Winslow%
  \thanks{Department of Computer Science, Tufts University, \protect\url{awinslow@cs.tufts.edu}. 
  Supported in part by National Science Foundation grants CCF-0830734 and CBET-0941538.}
\and
  Damien Woods%
    \thanks{Computer Science, California Institute of Technology, \protect\url{woods@caltech.edu}. 
    Supported by National Science Foundation grants 0832824 (The Molecular Programming Project), CCF-1219274, and CCF-1162589.}
}
\date{}

\input{commands-tam.sty}

\begin{document}

\maketitle
\vspace{-6ex}
\begin{abstract}
We prove a negative result on the power of a model of algorithmic self-assembly for which it has been notoriously difficult to find general techniques and results. Specifically, we prove that Winfree's abstract Tile Assembly Model, when restricted to use noncooperative tile binding, is not intrinsically universal. This stands in stark contrast to the recent result that, via cooperative binding, the abstract Tile Assembly Model is indeed intrinsically universal. Noncooperative self-assembly, also known as ``temperature 1'', is where tiles bind to each other if they match on one or more sides, whereas cooperative binding requires binding on multiple sides. Our result shows that the change from single- to multi-sided binding qualitatively improves the kinds of dynamics and behavior that these models of nanoscale self-assembly are capable of. Our lower bound on simulation power holds in both two and three dimensions; the latter being quite surprising given that three-dimensional noncooperative tile assembly systems simulate Turing machines.
On the positive side, we exhibit a three-dimensional noncooperative self-assembly tile set capable of simulating any two-dimensional \emph{noncooperative} self-assembly system.

Our negative result can be interpreted to mean that Turing universal algorithmic behavior in self-assembly does not imply the ability to simulate arbitrary algorithmic self-assembly processes.
\end{abstract}

%

\section{Introduction}
\label{sec:intro}

Self-assembly is the process through which unorganized, simple, components automatically coalesce according to simple local rules to form some kind of target structure. It sounds simple, but the end result can be extraordinary. For example, researchers have been able to self-assemble a wide variety of structures experimentally at the nanoscale, such as regular arrays~\cite{WinLiuWenSee98}, fractal structures~\cite{RoPaWi04,FujHarParWinMur07}, smiling faces~\cite{rothemund2006folding,wei2012complex}, DNA tweezers~\cite{yurke2000dna}, logic circuits~\cite{seelig2006enzyme,qian2011scaling}, neural networks~\cite{qian2011neural}, and molecular robots\cite{DNARobotNature2010}. These examples are fundamental because they demonstrate that self-assembly can, in principle, be used to manufacture specialized geometrical, mechanical and computational objects at the nanoscale. Potential future applications of nanoscale self-assembly include the production of smaller, more efficient microprocessors and medical technologies that are capable of diagnosing and even treating disease at the cellular level.

Controlling nanoscale self-assembly for the purposes of manufacturing atomically precise components will require a bottom-up, hands-off strategy. In other words, the self-assembling units themselves will have to be ``programmed'' to direct themselves to do the right thing--efficiently and correctly.  Molecular self-assembly is rapidly becoming  a ubiquitous engineering paradigm, and we need to develop a theory to inform
us of its algorithmic capabilities and ultimate limitations. %

In 1998, Erik Winfree \cite{Winf98} introduced the abstract Tile Assembly Model (aTAM), an over-simplified discrete mathematical model of algorithmic DNA nanoscale self-assembly pioneered by Seeman \cite{Seem82}. The aTAM is an  asynchronous  nondeterministic cellular automaton, that models crystal growth processes. Put another way, the aTAM essentially augments classical Wang tiling \cite{Wang61} with a mechanism for sequential ``growth'' of a tiling (in Wang tiling, only the existence of a valid, mismatch-free tiling is considered and not the order of tile placement). In the aTAM, the fundamental components are un-rotatable, but translatable square or cube ``tile types'' whose sides are labeled with %
``glue'' colors and integer ``strengths''. Two tiles that are placed next to each other \emph{interact} if the glue colors on their abutting sides match, and they \emph{bind} if the strengths on their abutting sides match and sum to at least a certain (integer) ``temperature''. Self-assembly starts from a ``seed'' tile type and proceeds nondeterministically and asynchronously as tiles bind to the seed-containing-assembly. Despite its deliberate over-simplification, the aTAM is a computationally expressive model. For example, by using  cooperative binding (that is, binding of tiles on two or more sides), Winfree \cite{Winf98} proved that it is Turing universal, which implies that self-assembly can be directed by a computer program. Here, we study  noncooperative binding.

Tile self-assembly in which tiles may be placed in a {\em noncooperative} fashion is colloquially referred to as ``temperature-1 self-assembly''.  Despite the esoteric name, this is a fundamental and ubiquitous form of growth: it refers to growth from {\em growing and branching tips} in  %
Euclidian space where each new tile is added if it can match on at least {\em one side}.
It has been known for some time that a more general form of growth where some of the tiles must match on two or more sides,  i.e.\  {\em cooperative} growth, leads to highly non-trivial behavior: arbitrary Turing machine simulation~\cite{RotWin00, jCCSA}, efficient production of $n \times n$ squares and other simple shapes using $\Theta(\log n/\log \log n)$ tile types~\cite{AdChGoHu01}, efficient production of arbitrary finite connected shapes using a number of tile types that is within a log factor of the Kolmogorov complexity of the shape~\cite{SolWin07}, and even intrinsic universality: the existence of a single tile set that simulates arbitrary tile assembly systems~\cite{IUSA}.
Until now, it was not known whether or not two-dimensional noncooperative binding has these capabilities without possibility of error, although in all cases the answer has been conjectured to be negative~\cite{RotWin00, Doty-2011, Reif-2012, Cook-2011, Manuch-2010, Patitz-2011}.  %
Our main result is such a negative result. %
Simply put, there is no noncooperative tile set that simulates all other tile assembly systems.

The topic of intrinsic universality, with its tight notion of simulation, has given rise to a rich theory in the field of cellular automata~\cite{bulkingI,bulkingII,arrighi2012intrinsic,Ollinger08}, and indeed has also been studied in Wang tiling~\cite{LafitteW07,LafitteW08,LafitteW09} and  tile self-assembly~\cite{USA, IUSA, 2HAMIU}.  Recently, the aTAM has been shown to be intrinsically universal~\cite{IUSA}, meaning that there is a single set of tiles~$U$ that works at temperature 2, and when appropriately initialized, is capable of simulating the behavior of an arbitrary aTAM tile assembly system. Modulo rescaling, this single tile set $U$ represents the full power and expressivity of the entire aTAM model, at any temperature. Indeed, Demaine et al~\cite{Demaine-2012}  apply this to show that there is a single (rotatable, translatable)  polygonal tile that can simulate any tile assembly system or Wang plane tiling system.   The restricted ``locally consistent'' aTAM  also exhibits intrinsic universality~\cite{USA}.
More recently, it has been shown that  the two-handed model of self-assembly (where large assemblies of tiles may come together in a single step) is not intrinsically universal~\cite{2HAMIU}. However, the same paper shows that  for each ``temperature'' $\tau \in \{ 2,3,4, \ldots \}$ there is a tileset  that is intrinsically universal for the class of two-handed systems that work at temperature $\tau$~\cite{2HAMIU}, and that there is an  infinite hierarchy of classes of systems with each level strictly more powerful than the one below.  As has been done for cellular automata, intrinsic universality in self-assembly, with its well-defined and powerful notion of simulation, is becoming  a new  tool by which we can tease apart the computational power of self-assembly systems.

\subsection{Results}
We give an overview of our results, although a number of terms have not yet been formally defined. For %
definitions, see  Section~\ref{sec:prelims}.
Our main result states that in the standard   noncooperative model  (i.e.\ temperature-1 aTAM in 2D) there is no intrinsically universal tile set. The proof is contained in Section~\ref{sec:2D_temp1_not_aTAM_universal}.

\begin{theorem}\label{thm:not-iu}
There is no tile set $U$ such that $U$ is intrinsically universal at temperature 1 for the class of all aTAM tile assembly systems.
\end{theorem}
Our main result stands in stark contrast to the fact that if we permit cooperative binding (that is, temperature 2) then there is a universal tile set for the aTAM: %
\begin{nunumtheorem}[Doty, Lutz, Patitz, Schweller, Summers, Woods $\!\!\!\,$\cite{IUSA}] \label{thm:IUSA}
There is a tile set $U$ such that $U$  is intrinsically universal at temperature 2 for the class of all aTAM tile assembly
systems.
\end{nunumtheorem}
This proves that noncooperative  systems can not simulate  cooperative systems, and shows that temperature 1 systems are provably weaker than temperature 2 systems in terms of their ability to simulate structure and dynamics. The same proof from Section~\ref{sec:2D_temp1_not_aTAM_universal}  also works in 3D: %
\begin{theorem}\label{thm:3d-not-iu}
There is no 3D tile set $U$ such that $U$ is intrinsically universal at temperature 1 for the class of all  aTAM tile assembly systems.
\end{theorem}
The latter negative result is interesting in how it stands in contrast to the known result that 3D temperature 1 can indeed simulate arbitrary algorithms:
\begin{nunumtheorem}[Cook, Fu, Schweller \cite{Cook-2011}] \label{thm:3d-t1-computation}
For each Turing machine $M$ and input~$x$ there exists a 3D temperature 1 tile assembly system $\mathcal{T}_{M,x}$ that simulates  the computation of  $M$ on $x$.
\end{nunumtheorem}

So, the process of tile assembly can be simulated by a (Turing machine) algorithm, and  3D temperature-1 can simulate arbitrary algorithms, yet 3D temperature-1 can not simulate self-assembly in a way that preserves structure and dynamics. This result essentially says that in a noncooperative growth-based setting, the ability to simulate arbitrary algorithms does not confer the ability to simulate arbitrary algorithmic tile-based growth dynamics.

\subsubsection{Positive results}
Our negative results should be contrasted with our positive result, which is proved in Section~\ref{sec:3D_temp1_simulates_2D_temp1}. We find that 3D noncooperative tile assembly can in fact simulate 2D noncooperative tile assembly, in other words, 3D temperature-1 simulates 2D temperature-1.

\begin{theorem}\label{thm:3D-sim-2D-simple-statement}
There is a 3D tile set $U$ such that $U$ is intrinsically universal at temperature 1 for the class of all 2D aTAM tile assembly systems.
\end{theorem}
Finally, we conjecture the following:
\begin{conjecture}\label{conj:no-2d-sim-2d}
There is no 2D tile set $U$ such that $U$ is intrinsically universal at temperature 1 for the class of all 2D aTAM temperature 1 tile assembly systems.
\end{conjecture}

\subsubsection{Other results}
The proof of Theorems~\ref{thm:not-iu} and \ref{thm:3d-not-iu}, also holds for the restricted class of ``locally consistent'' aTAM systems~\cite{USA}. In~\cite{USA} it was shown that there is a locally consistent tile set that is intrinsically universal at temperature 2 for all locally constant systems. Here we show that temperature 1 can not even simulate this restricted class of systems (proof: the TAS~$\mathcal{T}$ shown to be un-simulatable at temperature 1 in the proof of Theorem~\ref{thm:not-iu} is locally  consistent):
\begin{theorem}\label{thm:locally-consistent}
There is no tile set $U$ such that $U$ is intrinsically universal at temperature 1 for the class of all locally consistent aTAM tile assembly systems.
\end{theorem}

Intrinisic universality uses a strong notion of simulation where the simulator is a single tile set that simulates all tile assembly systems from some class. A weaker form of simulation is where for each tile assembly system $\mathcal{T}$ from some class, there exists a simulator  tile assembly system $\mathcal{T}'$ (from another class), that simulates  $\mathcal{T}$ (see, e.g., \cite{AKKRS09, versus, Demaine-2012}). Our proof shows that even this weaker form of simulation of temperature-2 is impossible at temperature 1:
\begin{theorem}\label{thm:no-forall-exists-sim}
There is a 2D temperature-2 tile assembly system $\mathcal{T}$ that can not be simulated by any 2D, nor any 3D, temperature 1 tile assembly system.
\end{theorem}
The proof of this is the same as the proofs of Theorem~\ref{thm:3d-not-iu},  and given in Section~\ref{sec:2D_temp1_not_aTAM_universal}.\footnote{To see that the same proof applies, note that Section~\ref{sec:2D_temp1_not_aTAM_universal} defines a specific temperature 2 tile assembly system $\mathcal{T}$, and shows  that there is no  temperature-1 simulator for $\mathcal{T}$.}

\subsection{Key technical ideas and methods}
One of the main challenges with proving negative results about 2D temperature~1 self-assembly comes from the intuition that, although the assemblies produced at temperature-1 often look ``obviously simple'' (they are a collection of simple paths, possibly with repeating tile types), it seems extremely difficult to prove this.
This is because it is easy to overlook geometry and quickly become seduced into believing that, as a result of the noncooperative nature of temperature-1 self-assembly, it must always be possible to indefinitely repeat (or ``pump'') sub-paths of tiles that begin and end with the same tile type. However, it is easy to construct an example of a 2D temperature-1 self-assembly system that uniquely produces a final structure, which contains at least one sub-path that begins and ends with the same tile type but the sub-path can not be pumped indefinitely because it gets ``blocked'' by previous portions of the path.
Could a long growth path that blocks itself, but branches just before doing so, simulate meaningful computation? Surprisingly, both the  the 2D low-error, and 3D no-error, temperature-1 Turing machine simulations in~\cite{Cook-2011} iterate exactly this idea, over and over, along with some clever geometric tricks. Our result here shows that neither this, nor any trick, will suffice to show that 2D nor 3D temperature-1 simulates aTAM tile self-assembly. %

To show this limitation on temperature-1, we first prove Lemma~\ref{lem:windowmovie} that gives a sufficient condition for taking any two assemblies, at any temperature $\geq 1$, and ``splicing'' them together to create a new valid assembly.  This gives a kind of strong pumping lemma for self-assembly. This lemma generalizes Theorem 3.1 of~\cite{Aggarwal-2005}, which was (a) proven for a more  restrictive scenario where the assemblies are contained in long and thin rectangles, and (b) works only for pumping a positive number of times---ours works for negative pumping (i.e. shrinking/splicing out) also.%

Armed with this lemma, we then give an example, very simple, temperature 2 tile assembly system $\mathcal{T}$ that uses cooperative binding (binding on 2 sides) in exactly one tile position, with all other bonds being noncooperative. We show that any claimed  temperature-1 simulation of this system must fail, and the place it fails is at the location where it should simulate cooperative binding. Any claimed simulator tile set is free to choose to use arbitrary scaling and a complicated-looking seed assembly, and may have a large (but constant) number of tile types; nevertheless we can use our pumping lemma to splice out parts of the simulation and trick it into exposing its inability to simulate cooperation. The proof is given in Section~\ref{sec:2D_temp1_not_aTAM_universal}: it works in both 2D and 3D which  gives Theorems~\ref{thm:not-iu} and~\ref{thm:3d-not-iu}. In the proof, since $\mathcal{T}$ is locally consistent, we get also Theorem~\ref{thm:locally-consistent}, and since we exhibited  a specific $\mathcal{T}$ that can not be simulated we get Theorem~\ref{thm:no-forall-exists-sim}.

In Section~\ref{sec:3D_temp1_simulates_2D_temp1} we show that the 3D temperature-1 aTAM can indeed simulate the 2D temperature-1 aTAM. The construction makes extensive use of the fact that in 3D, a closed curve does not necessarily partition the space into two parts.  It repeatedly  uses the third dimension as a means of sidestepping the limitations of planarity, and for ``stepping up and over'' locations reserved for future growth, then ``stepping down'' to place  blocking tiles which will later block specific paths, and then returning to continue growth  %
along a path which will eventually read this geometric blocking information. 
Similar blocking was used by %
Cook, Fu, and Schweller~\cite{Cook-2011}.  However, their construction consists of one single non-blocked  path, with many tiny blocked branches. %
Our construction simulates the multiple, often  independent,  paths of the simulated system %
by using many paths, each of which has many tiny branches that all get blocked, except for one.  This forces the construction to correctly handle a variety of timing issues related to the growth of the assembly, always ensuring that any needed blocking tiles \emph{must} be placed before the path which will ``read'' them cane form, and also to correctly deal with all possible situations where divergent paths (i.e.\ those simulating the independent additions of separate tiles) may later converge on a location.  This is dealt with using a ``competition'' scheme similar to that in \cite{USA} and \cite{IUSA}.

\subsection{Prior work on noncooperative binding}
Many examples (referenced above) testify that cooperative binding in tile self-assembly is sufficient for the self-assembly of computationally and geometrically interesting shapes and patterns.
But is it necessary? In other words, is cooperative binding more powerful than noncooperative binding?

Unfortunately and frustratingly, few general techniques exist for proving lower bounds in 2D temperature-1 self-assembly. However, there are some nice examples that begin to expose its the limitations. For instance, Rothemund and Winfree~\cite{RotWin00} proved that the  number of unique tile types required to uniquely self-assemble a \emph{fully-connected} $n \times n$ square in 2D at temperature-1 is $\geq n^2$ and conjectured that, in general, $2n-1$ unique tile types are necessary to uniquely self-assemble $n \times n$ squares at temperature-1. Manuch et al.~\cite{Manuch-2010} proved that the minimum number of unique tile types required to uniquely self-assemble an $n \times n$ square in 2D, at temperature-1, with {\em no glue mismatches}, is $2n - 1$. Note that the latter result does not assume a fully-connected terminal structure, whereas the former does. Doty,  Patitz and Summers~\cite{Doty-2011} formalized a notion of ``pumpability'' in temperature-1 self-assembly: a 2D temperature-1 self-assembly system that uniquely produces an infinite structure is ``pumpable'', if for every sufficiently long path of tiles, it is always possible to find at least one infinitely repeatable sub-path of tiles along this path (although not every sub-path that begins and ends with the same tile type may be infinitely repeatable).  They conjecture that all 
2D temperature-1 tile systems that uniquely produce some final structure are pumpable, 
and under the assumption of pumpbility  they prove that %
 the shape or pattern it produced is necessarily ``simple'' in the sense of Presburger arithmetic~\cite{Presburger30}. However, their  conjecture %
remains unproven.

\subsection{Prior work on intrinsic universality}
Intrinsic universality uses a strict notion of simulation, where the simulator preserves the dynamics of the simulated system, modulo a constant-sized (block) rescaling. In particular, an intrinsically universal cellular automata is one where its space-time diagrams contain (via a representation function) those of {\em any} simulated cellular automaton: where (in 1D) a single cell in the simulated automaton is represented by an $m \times t$ block in the simulator. Despite this strong requirement, intrinsically universal cellular automata were shown to be very common in some natural classes of rules~\cite{BoyerT09} and there are examples with very small programs (rules)~\cite{ollinger-fourstates}.
The idea that intrinsic universality could facilitate the finding of lower bounds and negative results was conjectured, for example in~\cite{ollinger-fourstates}, and a general method was proposed in~\cite{goles-communicationcomplexity}. Since then, intrinsic universality, and in particular communication complexity theory, have been used as general tools to show negative results on cellular automata~\cite{goles-communicationcomplexity,NesmeT11,BricenoR13,ChaccMR11}.

The notion of simulation we use can be thought of as a reduction between systems,  however  it is stronger than usual reductions defined via
algorithmic resource constraints (time, space, even constant circuit depth, etc.). %
For computational models it is often difficult to prove negative results  separating computational power, however, our strict notion of simulation shifts the difficulty from proving hardness results to proving simulation (or
\emph{completeness}) results.
However, now that we have examples of intrinsically universal tile sets~$U$~\cite{IUSA, USA, 2HAMIU} we know that arbitrary  ``tile programs'' can be written, analyzed and compiled into such $U$; it captures everything (modulo rescaling).  Not only that,  we claim that  our notion of simulation is a powerful tool because %
we have gained the ability to prove lower bounds and impossibility results, as this paper shows

\section{Preliminaries}\label{sec:prelims}

\subsection{Informal description of the abstract Tile Assembly Model}
\label{sec-tam-informal}

This section gives a brief informal sketch of the abstract Tile Assembly Model (aTAM). %
See Section~\ref{sec-tam-formal} for a formal definition of the aTAM. In this section, we define the 2D aTAM, whereas in Section~\ref{sec-tam-formal} we formulate the $d$-dimensional aTAM. For purposes of notational convenience, throughout this paper we will use the term ``aTAM'' will refer to the 2D aTAM.

A \emph{tile type} is a unit square with four sides, each consisting of a \emph{glue label}, often represented as a finite string, and a nonnegative integer \emph{strength}. A glue~$g$ that appears on multiple tiles (or sides) always has the same strength~$s_g$. %
There are a finite set $T$ of tile types, but an infinite number of copies of each tile type, with each copy being referred to as a \emph{tile}. An \emph{assembly}
is a positioning of tiles on the integer lattice $\Z^2$, described  formally as a partial function $\alpha:\Z^2 \dashrightarrow T$. %
Let $\mathcal{A}^T$ denote the set of all assemblies of tiles from $T$, and let $\mathcal{A}^T_{< \infty}$ denote the set of finite assemblies of tiles from $T$.
We write $\alpha \sqsubseteq \beta$ to denote that $\alpha$ is a \emph{subassembly} of $\beta$, which means that $\dom\alpha \subseteq \dom\beta$ and $\alpha(p)=\beta(p)$ for all points $p\in\dom\alpha$.
Two adjacent tiles in an assembly \emph{interact}, or are \emph{attached}, if the glue labels on their abutting sides are equal and have positive strength. %
Each assembly induces a \emph{binding graph}, a grid graph whose vertices are tiles, with an edge between two tiles if they interact.
The assembly is \emph{$\tau$-stable} if every cut of its binding graph has strength at least~$\tau$, where the strength   of a cut is the sum of all of the individual glue strengths in the cut.

A \emph{tile assembly system} (TAS) is a triple $\calT = (T,\sigma,\tau)$, where $T$ is a finite set of tile types, $\sigma:\Z^2 \dashrightarrow T$ is a finite, $\tau$-stable \emph{seed assembly},
and $\tau$ is the \emph{temperature}.
An assembly $\alpha$ is \emph{producible} if either $\alpha = \sigma$ or if $\beta$ is a producible assembly and $\alpha$ can be obtained from $\beta$ by the stable binding of a single tile.
In this case we write $\beta\to_1^\calT \alpha$ (to mean~$\alpha$ is producible from $\beta$ by the attachment of one tile), and we write $\beta\to^\calT \alpha$ if $\beta \to_1^{\calT*} \alpha$ (to mean $\alpha$ is producible from $\beta$ by the attachment of zero or more tiles).
When $\calT$ is clear from context, we may write $\to_1$ and $\to$ instead.
We let $\prodasm{\calT}$ denote the set of producible assemblies of $\calT$.
An assembly is \emph{terminal} if no tile can be $\tau$-stably attached to it.
We let   $\termasm{\calT} \subseteq \prodasm{\calT}$ denote  the set of producible, terminal assemblies of $\calT$.
A TAS $\calT$ is \emph{directed} if $|\termasm{\calT}| = 1$. Hence, although a directed system may be nondeterministic in terms of the order of tile placements,  it is deterministic in the sense that exactly one terminal assembly is producible (this is analogous to the notion of {\em confluence} in rewriting systems).
Since the behavior of a TAS $\calT=(T,\sigma,\tau)$ is unchanged if every glue with strength greater than $\tau$ is changed to have strength exactly $\tau$, we assume  that all glue strengths are in the set $\{0, 1, \ldots , \tau\}$.

\subsection{Formal description of the abstract Tile Assembly Model}
\label{sec-tam-formal}

This section gives a formal definition of the abstract Tile Assembly Model (aTAM)~\cite{Winfree-1998}. For readers unfamiliar with the aTAM, Section~\ref{sec-tam-informal} contains a less formal overview and~\cite{RotWin00} gives an excellent introduction to the model.

Fix an alphabet $\Sigma$.
$\Sigma^*$ is the set of finite strings over $\Sigma$.
$\Z$, $\Z^+$, and $\N$ denote the set of integers, positive integers, and nonnegative integers, respectively.
Let $d \in \{2,3\}$.
Given $V \subseteq \Z^d$, the \emph{full grid graph} of $V$ is the undirected graph $\fullgridgraph_V=(V,E)$, %
and for all $\vec{x} = \left(x_0, \ldots, x_{d-1}\right), \vec{y} = \left(y_0, \ldots, y_{d-1}\right)\in V$, $\left\{\vec{x},\vec{y}\right\} \in E \iff \| \vec{x} - \vec{y}\| = 1$; i.e., if and only if $\vec{x}$ and $\vec{y}$ are adjacent on the $d$-dimensional integer Cartesian space.

A $d$-dimensional \emph{tile type} is a tuple $t \in (\Sigma^* \times \N)^{2d}$; e.g., a unit square (or cube) with four (or six) sides listed in some standardized order, each side having a \emph{glue} $g \in \Sigma^* \times \N$ consisting of a finite string \emph{label} and nonnegative integer \emph{strength}. From this point on, a \emph{tile} will refer to either a 2D square or 3D cube tile type.
We assume a finite set of tile types, but an infinite number of copies of each tile type, each copy referred to as a \emph{tile}. A $d$-dimensional tile set is a set of $d$-dimensional tile types and is written as $d$-$T$. A tile set~$T$ is a set of $d$-dimensional tile types for some $d \in \{2,3\}$.
A $d$-{\em configuration} is a (possibly empty) arrangement of tiles on the integer lattice $\Z^d$, i.e., a partial function $\alpha:\Z^d \dashrightarrow T$. A configuration $\alpha$ is a $d$-configuration for some $d \in \{2,3\}$.
A $d$-\emph{assembly} is a connected non-empty configuration, i.e., a partial function $\alpha:\Z^d \dashrightarrow T$ such that $\fullgridgraph_{\dom \alpha}$ is connected and $\dom \alpha \neq \emptyset$. An assembly is a $d$-assembly for some $d \in \{2,3\}$.

Let $\mathcal{A}^T$ denote the set of all assemblies of tiles from $T$, and let $\mathcal{A}^T_{< \infty}$ denote the set of finite assemblies of tiles from $T$.
The \emph{shape} $S_\alpha \subseteq \Z^d$ of $\alpha$ is $\dom \alpha$.
Two adjacent tiles in an assembly \emph{interact}, or are \emph{attached}, if the glues on their abutting sides are equal (in both label and strength) and have positive strength.
Each assembly $\alpha$ induces a \emph{binding graph} $\bindinggraph_\alpha$, a grid graph whose vertices are positions occupied by tiles, with an edge between two vertices if the tiles at those vertices interact.
Given $\tau\in\Z^+$, $\alpha$ is \emph{$\tau$-stable} if every cut of~$\bindinggraph_\alpha$ has weight at least $\tau$, where the weight of an edge is the strength of the glue it represents. %
When $\tau$ is clear from context, we say $\alpha$ is \emph{stable}.
Given two assemblies $\alpha,\beta$, we say $\alpha$ is a \emph{subassembly} of $\beta$, and we write $\alpha \sqsubseteq \beta$, if $S_\alpha \subseteq S_\beta$ and, for all points $p \in S_\alpha$, $\alpha(p) = \beta(p)$.

A $d$-dimensional \emph{tile assembly system} ($d$-TAS) is a triple $d$-$\mathcal{T} = (d\textrm{-}T,\sigma,\tau)$, where $d$-$T$ is a finite set of $d$-dimensional tile types, $\sigma:\Z^d \dashrightarrow T$ is the finite, $\tau$-stable, $d$-dimensional \emph{seed assembly}, and $\tau\in\Z^+$ is the \emph{temperature}. The triple $\mathcal{T} = (T,\sigma,\tau)$ is a TAS if it is is a $d$-TAS for some $d \in \{2,3\}$.
Given two $\tau$-stable assemblies $\alpha,\beta$, we write $\alpha \to_1^{\mathcal{T}} \beta$ if $\alpha \sqsubseteq \beta$ and $|S_\beta \setminus S_\alpha| = 1$. In this case we say $\alpha$ \emph{$\mathcal{T}$-produces $\beta$ in one step}. %
If $\alpha \to_1^{\mathcal{T}} \beta$, $ S_\beta \setminus S_\alpha=\{p\}$, and $t=\beta(p)$, we write $\beta = \alpha + (p \mapsto t)$.
The \emph{$\mathcal{T}$-frontier} of $\alpha$ is the set $\partial^\mathcal{T} \alpha = \bigcup_{\alpha \to_1^\mathcal{T} \beta} S_\beta \setminus S_\alpha$, the set of empty locations at which a tile could stably attach to $\alpha$. The \emph{$t$-frontier} $\partial_t \alpha \subseteq \partial \alpha$ of $\alpha$ is the set $\setr{p\in\partial \alpha}{\alpha \to_1^\mathcal{T} \beta \text{ and } \beta(p)=t}.$ %

A sequence of $k\in\Z^+ \cup \{\infty\}$ assemblies $\alpha_0,\alpha_1,\ldots$ over $\mathcal{A}^T$ is a \emph{$\mathcal{T}$-assembly sequence} if, for all $1 \leq i < k$, $\alpha_{i-1} \to_1^\mathcal{T} \alpha_{i}$.
The {\em result} of an assembly sequence is the unique limiting assembly (for a finite sequence, this is the final assembly in the sequence).

We write $\alpha \to^\mathcal{T} \beta$, and we say $\alpha$ \emph{$\mathcal{T}$-produces} $\beta$ (in 0 or more steps) if there is a $\mathcal{T}$-assembly sequence $\alpha_0,\alpha_1,\ldots$ of length $k = |S_\beta \setminus S_\alpha| + 1$ such that
\begin{enumerate}
\item $\alpha = \alpha_0$,
\item $S_\beta = \bigcup_{0 \leq i < k} S_{\alpha_i}$, and
\item for all $0 \leq i < k$, $\alpha_{i} \sqsubseteq \beta$.
\end{enumerate}
If $k$ is finite then it is routine to verify that $\beta = \alpha_{k-1}$. 
We say $\alpha$ is \emph{$\mathcal{T}$-producible} if $\sigma \to^\mathcal{T} \alpha$, and we write $\prodasm{\mathcal{T}}$ to denote the set of $\mathcal{T}$-producible assemblies. The relation $\to^\mathcal{T}$ is a partial order on $\prodasm{\mathcal{T}}$ \cite{Rothemund-2001,Lathrop-2009}.
An assembly $\alpha$ is \emph{$\mathcal{T}$-terminal} if $\alpha$ is $\tau$-stable and $\partial^\mathcal{T} \alpha=\emptyset$.
We write $\termasm{\mathcal{T}} \subseteq \prodasm{\mathcal{T}}$ to denote the set of $\mathcal{T}$-producible, $\mathcal{T}$-terminal assemblies. If $|\termasm{\mathcal{T}}| = 1$ then  $\mathcal{T}$ is said to be {\em directed}.

When $\mathcal{T}$ is clear from context, we may omit $\mathcal{T}$ from the notation above and instead write
$\to_1$,
$\to$,
$\partial \alpha$, %
\emph{assembly sequence},
\emph{produces},
\emph{producible}, and
\emph{terminal}.

\subsection{Simulation definition}
\label{sec:simulation_def}

To state our main result, we must formally define what it means for one TAS to ``simulate'' another.  %
The following definitions improve the  presentation of those in~\cite{IUSA}, and correct a subtle error there.\footnote{Roughly speaking, Definition~\ref{def-s-models-t} uses an existential  quantifier, whereas the version in~\cite{IUSA} used a universal quantifier. This correction still captures the intention in~\cite{IUSA}, and it 
 actually strengthens our main results (i.e.\ our negative results:  Theorems~\ref{thm:not-iu},~\ref{thm:3d-not-iu}, \ref{thm:locally-consistent}, and~\ref{thm:no-forall-exists-sim}) without invalidating the positive result here (Theorem~\ref{thm:3D-sim-2D-simple-statement}) nor  that  in~\cite{IUSA}.}

From this point on, let $T$ be a $d$-dimensional tile set, and let %
$m\in\Z^+$.
An \emph{$m$-block supertile} over $T$ is a partial function $\alpha : \Z_m^d \dashrightarrow T$, where $\Z_m = \{0,1,\ldots,m-1\}$. Note that the dimension of the $m$-block is implicitly defined by the dimension of $T$.
Let $B^T_m$ be the set of all $m$-block supertiles over $T$.
The $m$-block with no domain is said to be $\emph{empty}$.
For a general assembly $\alpha:\Z^d \dashrightarrow T$ and $(x_0,\ldots x_{d-1})\in\Z^d$, define $\alpha^m_{x_0,\ldots x_{d-1}}$ to be the $m$-block supertile defined by $\alpha^m_{x_0,\ldots, x_{d-1}}(i_0,\ldots, i_{d-1}) = \alpha(mx_0+i_0,\ldots, mx_{d-1}+i_{d-1})$ for $0 \leq i_0, \ldots, i_{d-1}< m$.
For some tile set $S$ of dimension $d' \geq d$, a partial function $R: B^{S}_m \dashrightarrow T$ is said to be a \emph{valid $m$-block supertile representation} from $S$ to $T$ if for any $\alpha,\beta \in B^{S}_m$ such that $\alpha \sqsubseteq \beta$ and $\alpha \in \dom R$, then $R(\alpha) = R(\beta)$.
Let  $d' \in \{ 2,3 \}$ and $d \in \{ d' -1, d' \} $.  
Let $f: \Z^{d'} \rightarrow \Z^{d}$, where $f(x_0,\ldots,x_{d'-1}) = (x_0,\ldots,x_{d'-1})$ if $d'=d$ and $f(x_0,\ldots,x_{d'-1}) = (x_0,\ldots,x_{d-1},0)$ if $d = d' -1$, and undefined otherwise.
For a given valid $m$-block supertile representation function $R$ from tile set~$S$ to tile set $T$, define the \emph{assembly representation function}\footnote{Note that $R^*$ is a total function since every assembly of $S$ represents \emph{some} assembly of~$T$; the functions $R$ and $\alpha$ are partial to allow undefined points to represent empty space.}  $R^*: \mathcal{A}^{S} \rightarrow \mathcal{A}^T$ such that $R^*(\alpha') = \alpha$ if and only if $\alpha(x_0,\ldots, x_{d-1}) = R\left(\alpha'^m_{x_0,\ldots, x_{d'-1}}\right)$ for all $(x_0,\ldots x_{d'-1}) \in \Z^{d'-1}$. 
For an assembly $\alpha' \in \mathcal{A}^{S}$ such that $R(\alpha') = \alpha$, $\alpha'$ is said to map \emph{cleanly} to $\alpha \in \mathcal{A}^T$ under $R^*$ if for all non empty blocks $\alpha'^m_{x_0,\ldots, x_{d'-1}}$, $(f(x_0,\ldots,x_{d'-1})+f(u_0,\ldots,u_{d'-1})) \in \dom \alpha$ for some $u_0,\ldots, u_{d'-1} \in \{-1,0,1\}$ such that $u_0^2 + \cdots + u_{d'-1}^2 \leq 1$, or if $\alpha'$ has at most one non-empty $m$-block $\alpha^m_{0, \ldots, 0}$.
In other words, $\alpha'$ may have tiles on supertile blocks representing empty space in $\alpha$, but only if that position is adjacent to a tile in $\alpha$.  We call such growth ``around the edges'' of $\alpha'$ \emph{fuzz} and thus restrict it to be adjacent to only valid supertiles, but not diagonally adjacent (i.e.\ we do not permit \emph{diagonal fuzz}).

In the following definitions, let $\mathcal{T} = \left(T,\sigma_T,\tau_T\right)$ be a $d$-TAS for $d \in \{2,3\}$, let $\mathcal{S} = \left(S,\sigma_S,\tau_S\right)$ be a $d'$-TAS for $d' \geq d$, and let $R$ be an $m$-block representation function $R:B^S_m \rightarrow T$.

\begin{definition}
\label{def-equiv-prod} We say that $\mathcal{S}$ and $\mathcal{T}$ have \emph{equivalent productions} (under $R$), and we write $\mathcal{S} \Leftrightarrow \mathcal{T}$ if the following conditions hold:
\begin{enumerate}
        \item $\left\{R^*(\alpha') | \alpha' \in \prodasm{\mathcal{S}}\right\} = \prodasm{\mathcal{T}}$.
        \item $\left\{R^*(\alpha') | \alpha' \in \termasm{\mathcal{S}}\right\} = \termasm{\mathcal{T}}$.
        \item For all $\alpha'\in \prodasm{\mathcal{S}}$, $\alpha'$ maps cleanly to $R^*(\alpha')$.
\end{enumerate}
\end{definition}

\begin{definition}
\label{def-t-follows-s} We say that $\mathcal{T}$ \emph{follows} $\mathcal{S}$ (under $R$), and we write $\mathcal{T} \dashv_R \mathcal{S}$ if $\alpha' \rightarrow^\mathcal{S} \beta'$, for some $\alpha',\beta' \in \prodasm{\mathcal{S}}$, implies that $R^*(\alpha') \to^\mathcal{T} R^*(\beta')$.
\end{definition}

\begin{definition}
\label{def-s-models-t} We say that $\mathcal{S}$ \emph{models} $\mathcal{T}$ (under $R$), and we write $\mathcal{S} \models_R \mathcal{T}$, if for every $\alpha \in \prodasm{\mathcal{T}}$, there exists $\Pi \subset \prodasm{\mathcal{S}}$ where $R^*(\alpha') = \alpha$ for all $\alpha' \in \Pi$, such that, for every $\beta \in \prodasm{\mathcal{T}}$ where $\alpha \rightarrow^\mathcal{T} \beta$, (1) for every $\alpha' \in \Pi$ there exists $\beta' \in \prodasm{\mathcal{S}}$ where $R^*(\beta') = \beta$ and $\alpha' \rightarrow^\mathcal{S} \beta'$, and (2) for every $\alpha'' \in \prodasm{\mathcal{S}}$ where $\alpha'' \rightarrow^\mathcal{S} \beta'$, $\beta' \in \prodasm{\mathcal{S}}$, $R^*(\alpha'') = \alpha$, and $R^*(\beta') = \beta$, there exists $\alpha' \in \Pi$ such that $\alpha' \rightarrow^\mathcal{S} \alpha''$.
\end{definition}

The previous definition essentially specifies that every time $\mathcal{S}$ simulates an assembly $\alpha \in \prodasm{\mathcal{T}}$, there must be at least one valid growth path in $\mathcal{S}$ for each of the possible next steps that $\mathcal{T}$ could make from $\alpha$ which results in an assembly in $\mathcal{S}$ that maps to that next step.

\begin{definition}
\label{def-s-simulates-t} We say that $\mathcal{S}$ \emph{simulates} $\mathcal{T}$ (under $R$) if $\mathcal{S} \Leftrightarrow_R \mathcal{T}$ (equivalent productions), $\mathcal{T} \dashv_R \mathcal{S}$ and $\mathcal{S} \models_R \mathcal{T}$ (equivalent dynamics).
\end{definition}

\newcommand{\REPL}{\mathsf{REPR}}
\newcommand{\frakC}{\mathfrak{C}}

Let $\REPL$ denote the set of all supertile representation functions (i.e., $m$-block supertile representation functions for some $m\in\Z^+$).
For some $d \in \{2,3\}$, let $\frakC$ be a class of $d$-dimensional tile assembly systems, and let $U$ be a $d'$-dimensional tile set for $d' \geq d$.
Note that every element of $\frakC$, $\REPL$, and $\mathcal{A}^U_{< \infty}$ is a finite object, hence can be represented in a suitable format for computation in some formal system such as Turing machines.
We say $U$ is \emph{intrinsically universal} for $\frakC$ \emph{at temperature} $\tau' \in \Z^+$ %
if there are computable functions $\mathcal{R}:\frakC \to \REPL$ and $S:\frakC \to \mathcal{A}^U_{< \infty}$ such that, for each $\mathcal{T} = (T,\sigma,\tau) \in \frakC$, there is a constant $m\in\N$ such that, letting $R = \mathcal{R}(\mathcal{T})$, $\sigma_\mathcal{T}=S(\mathcal{T})$, and $\mathcal{U}_\mathcal{T} = (U,\sigma_\mathcal{T},\tau')$, $\mathcal{U}_\mathcal{T}$ simulates $\mathcal{T}$ at scale $m$ and using supertile representation function $R$.
That is, $\mathcal{R}(\mathcal{T})$ outputs a representation function that interprets assemblies of $\mathcal{U}_\mathcal{T}$ as assemblies of $\mathcal{T}$, and $S(\mathcal{T})$ outputs the seed assembly used to program tiles from $U$ to represent the seed assembly of $\mathcal{T}$. We say that~$U$ is \emph{intrinsically universal} for $\frakC$ if it is intrinsically universal for $\frakC$ at some temperature $\tau'\in Z^+$. %

\section{Temperature 1 self-assembly is not intrinsically universal for the aTAM}
\label{sec:2D_temp1_not_aTAM_universal}

\newenvironment{theoremmynum}[2][Theorem]{\begin{trivlist}
\item[\hskip \labelsep {\bfseries #1}\hskip \labelsep {\bfseries #2}]}{\end{trivlist}}

 In this section we prove Theorem~\ref{thm:3d-not-iu} which is restated below.  The proof is for 3D systems, and so as an immediate corollary we get our main theorem, which is for standard 2D systems: Theorem~\ref{thm:not-iu}.  In the proof, our chosen temperature 2 tile assembly system $\mathcal{T}$ (that ``breaks'' any claimed simulator) is locally consistent, so we  also get Theorem~\ref{thm:locally-consistent}. Finally, since in the proof we exhibit  a specific $\mathcal{T}$ that can not be simulated, we  also get Theorem~\ref{thm:no-forall-exists-sim}.

\begin{theoremmynum}{\ref{thm:3d-not-iu}}%
There is no 3D tile set $U$ such that $U$ is intrinsically universal at temperature 1 for the class of all  aTAM tile assembly systems.
\end{theoremmynum}

\subsection{Proof overview of Theorem~\ref{thm:3d-not-iu}}

We prove Theorem~\ref{thm:3d-not-iu} by contradiction.
We suppose that there exists a universal tile set $U$  at temperature~$1$. We then choose a particular  temperature~$2$  tile assembly system~$\mathcal{T}$ and show that any simulation of~$\mathcal{T}$ by $U$ must build erroneous assemblies, failing to simulate both dynamics and production in Definition~\ref{def-s-simulates-t}.
The tile assembly system~$\mathcal{T}$ is illustrated in Figure~\ref{fig:construction}. A seed tile grows two ``arms'', each of arbitrary length, these arms each grow a ``finger'' and then try to {\em cooperatively} touch their fingers:  if they happened to choose arms of equal length the fingers  can cooperatively place a  keystone tile which leads to flagpole and flag tiles, if not growth stops. Clearly,~$\mathcal{T}$  is a very simple temperature~2 tile assembly system.\footnote{In fact,~$\mathcal{T}$ is locally consistent~\cite{USA}.}

Recall that given~$\mathcal{T}$, the simulator then gets to choose an arbitrary scale factor~$m \in \mathbb{N}$ and seed assembly~$\sigma_{\mathcal{T}}$ for the simulation.
Growing from the seed, the universal tile set~$U$ simulates a tile assembly system~$\mathcal{T}$ if and only if it simulates every possible sequence of tile additions producing a terminal assembly of $\mathcal{T}$ (at some $m$-scale blowup).
This includes all non-deterministic branches of assembly, such as the various lengths of the arms of $\mathcal{T}$.
Our approach is to take a valid simulation that simulates the placing of the keystone, and use it to show that the simulator must also produce another assembly that is invalid, i.e.\ it is not a simulation of $\mathcal{T}$ as defined in Definition~\ref{def-t-follows-s}.
In particular, when simulating the placing of the keystone, both arms should be the same length, however we show that $U$ must also construct keystone-placing assemblies that have arms of unequal lengths and so are not valid simulations.

In order to construct the invalid assembly, we prove a lemma (called the \emph{window movie lemma}, Lemma~\ref{lem:windowmovie}) that describes an operation for taking two producible assemblies, and combining them to create two new producible assemblies.
The lemma is rather general, and it applies to TASs of any temperature producing arbitrary (possibly infinite) assemblies.
The window movie lemma can be used as a pumping lemma (generalizing the technique used in the proof of Theorem 3.1 of~\cite{Aggarwal-2004,Aggarwal-2005}), or used to splice arbitrary assemblies together.

The proof finishes by invoking the fact that $U$ is a temperature 1 system at one key step: the placement of a specific tile by the simulator in (or near) the simulated keystone region.
At this point we apply the window movie lemma to the assembly sequence and splice together pieces of the valid assembly to produce a second, invalid assembly, essentially exposing the temperature 1 simulator as a charlatan that is (poorly) faking cooperation.
Our proof avoids the use of overly complicated case analyses that often arise when working with temperature 1 systems.

\subsection{Windows}
\label{sec:windows}

In order to prove $U$ produces invalid assemblies when simulating the aforementioned system, we develop a technique called \emph{window movies} for constructing additional producible assemblies of a tile set ($U$) and seed $\sigma$, given a some initial producible assembly.
Window movies share some similarities with the proof of Theorem 3.1 in~\cite{Aggarwal-2005}, which shows that thin rectangular assemblies can be ``pumped'' to create new producible assemblies of arbitrary length.
We strengthen the technique in~\cite{Aggarwal-2005} so that assemblies can also be ``pumped down'', generating producible assemblies \emph{smaller} than the original assembly.
Besides being useful for
Theorem~\ref{thm:not-iu}, this lemma gives a general method to combine assemblies together which might be useful elsewhere.

\begin{definition}
A window $w$ is a set of edges forming a cut-set in the infinite grid graph.
\end{definition}

Given a window $w$ and an assembly $\alpha$, a window that {\em intersects} $\alpha$ is a partioning of $\alpha$ into two  configurations (i.e.\ after being split into two parts, each part may or may not be  disconnected).
In this case we say that the window $w$ cuts the assembly $\alpha$ into two configurations $\alpha_L$ and~$\alpha_R$, where $\alpha = \alpha_L \cup \alpha_R$.
Given a window $w$, its translation by a vector $\vec{c}$,  written $ w + \vec{c}$ is simply the translation of each of $w$'s elements (edges) by~$\vec{c}$.
Examples of windows are shown in Figure~\ref{fig:window}.

\begin{figure}[ht]
\centering
\includegraphics[scale=1.0]{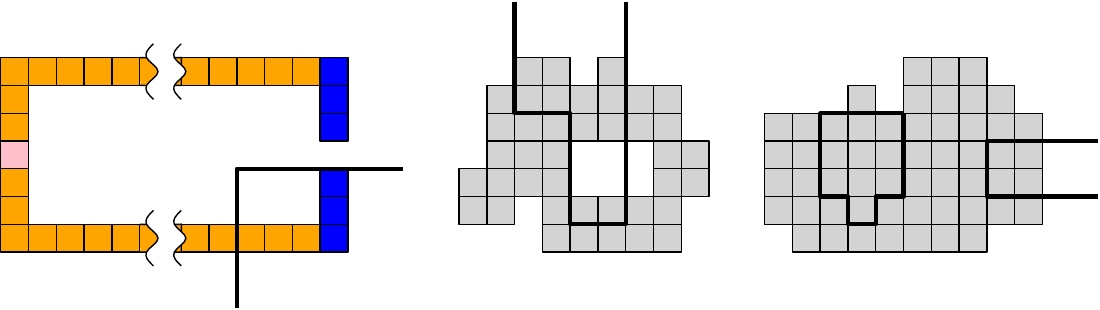}
\caption{Three examples of windows, shown as thick segments.
Each window partitions an assembly into two (not necessarily connected) configurations.}
\label{fig:window}
\end{figure}

For a window $w$ and an assembly sequence $\vec{\alpha}$, we define a window movie~$M$ to be the order of placement, position and glue type for each glue that appears along the window $w$ in an assembly sequence $\vec{\alpha}$.

\begin{definition}
Given an assembly sequence $\vec{\alpha}$ and a window $w$, the associated {\em window movie} is the maximal sequence $M_{\vec{\alpha},w} = (v_{0}, g_{0}) , (v_{1}, g_{1}), (v_{2}, g_{2}), \ldots$ of pairs of grid graph vertices $v_i$ and glues $g_i$, given by the order of the appearance of the glues along window $w$ in the assembly sequence $\vec{\alpha}$.
Furthermore, if $k$ glues appear along $w$ at the same instant (this happens upon placement of a tile which has multiple  sides  touching $w$) then these $k$ glues appear contiguously and are listed in lexicographical order of the unit vectors describing their orientation in $M_{\vec{\alpha},w}$.
\end{definition}

An example of a window movie is shown in Figure~\ref{fig:window_movie1}.

\begin{figure}[ht]
\centering
\includegraphics[scale=1.0]{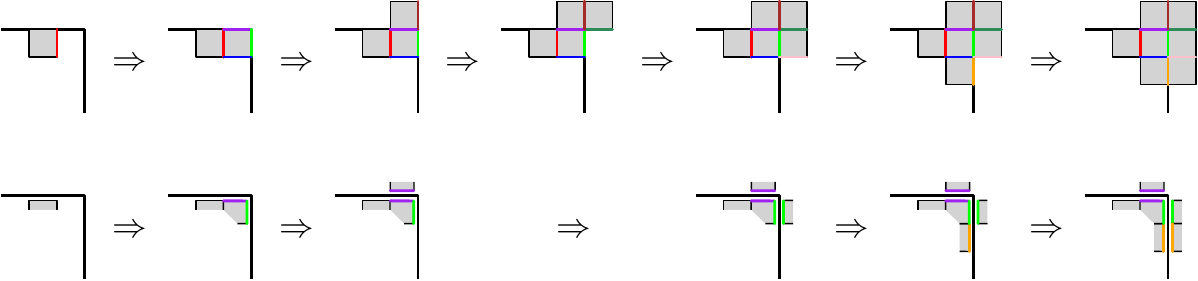}
\caption{Top: A window (thick line) and an assembly sequence along the window.
Bottom: The unique induced window movie.}
\label{fig:window_movie1}
\end{figure}

\pagebreak
\begin{lemma}[Window movie lemma]
\label{lem:windowmovie}
Let $\vec{\alpha} = (\alpha_i \mid 0 \leq i < l)$ and $\vec{\beta} = (\beta_i \mid 0 \leq i < m)$, with
$l,m\in\Z^+ \cup \{\infty\}$,
be assembly sequences in $\mathcal{T}$ with results $\alpha$ and $\beta$, respectively.
Let $w$ be a window that partitions~$\alpha$ into two configurations~$\alpha_L$ and $\alpha_R$, and $w' = w + \vec{c}$ be a translation of $w$ that partitions~$\beta$ into two configurations $\beta_L$ and $\beta_R$.
Furthermore, define $M_{\vec{\alpha},w}$, $M_{\vec{\beta},w'}$ to be the respective window movies for $\vec{\alpha},w$ and $\vec{\beta},w'$, and define $\alpha_L$, $\beta_L$ to be the subconfigurations of $\alpha$ and $\beta$ containing the seed tiles of $\alpha$ and $\beta$, respectively.
Then if $M_{\vec{\alpha},w} = M_{\vec{\beta},w'}$, it is the case that  the following two assemblies are also producible:
(1) the assembly $\alpha_L \beta'_R = \alpha_L \cup \beta'_R$ and
(2) the assembly $\beta'_L \alpha_R = \beta'_L \cup \alpha_R$, where $\beta'_L=\beta_L-\vec{c}$ and $\beta'_R=\beta_R-\vec{c}$.
\end{lemma}

Before proceeding, we first define some notation that will be useful for this section of the paper.

For an assembly sequence $\vec{\alpha} = (\alpha_i \mid 0 \leq i < l)$, we write $\left| \vec{\alpha} \right| = l$ (note that if $\vec{\alpha}$ is infinite, then $l = \infty$). We write $\vec{\alpha}[i]$ to denote $\vec{x} \mapsto t$, where $\vec{x}$ and~$t$ are such that $\alpha_{i+1} = \alpha_i + \left(\vec{x} \mapsto t\right)$, i.e., $\vec{\alpha}[i]$ is the placement  of tile type $t$ at position~$\vec{x}$, assuming that $\vec{x} \in \partial_{t}\alpha_i$. We define $\vec{\alpha} = \vec{\alpha} + \left(\vec{x} \mapsto t\right) = (\alpha_i \mid 0 \leq i < k + 1)$, where $\alpha_{k} = \alpha_{k-1} + \left(\vec{x} \mapsto t\right)$ if $\vec{x} \in \partial_{t}^{\tau}\alpha_i$ and undefined otherwise, assuming $\left| \vec{\alpha} \right| > 0$. Otherwise, if $\left| \vec{\alpha} \right| = 0$, then $\vec{\alpha} = \vec{\alpha} + \left(\vec{x} \mapsto t \right) = (\alpha_0)$, where $\alpha_0$ is the assembly such that $\alpha_0\left(\vec{x}\right) = t$ and is undefined at all other positions. This is our notation for appending steps to the assembly sequence $\vec{\alpha}$: to do so, we must specify a tile type $t$ to be placed at a given location $\vec{x} \in \partial_t\alpha_{i-1}$. If $\alpha_{i+1} = \alpha_i + \left(\vec{x} \mapsto t\right)$, then we write $Pos\left(\vec{\alpha}[i]\right) = \vec{x}$ and $Tile\left(\vec{\alpha}[i]\right) = t$. For a movie window $M = (v_0,g_0), (v_1,g_1), \ldots$, we write $M[k]$ to be the pair $\left(v_{k-1},g_{k-1}\right)$ in the enumeration of $M$ and $Pos\left(M[k]\right) = v_{k-1}$, where $v_{k-1}$ is a vertex of a grid graph.

\begin{proof}
We give a constructive proof by giving an algorithm for constructing an assembly sequence yielding $\alpha_L \beta_R'$. Let $\vec{\alpha}$ and $\vec{\beta}$ be the assembly sequences of $\alpha$ and $\beta$, respectively.
Intuitively, the algorithm performs a lossy merge of $\vec{\alpha}$ and $\vec{\beta}$, ignoring assembly sequence steps of $\vec{\alpha}$ (respectively, $\vec{\beta}$) that place tiles in $\alpha_R$ ($\beta_L'$).
Without loss of generality, and for notational simplicity,  let~$w$ be a window such that $M_{\vec{\alpha},w} = M_{\vec{\beta},w}$. In other words,  the common window movie of $\vec{\alpha}$ and $\vec{\beta}$ occur at the same location in the plane, and thus since $\vec{c} = \vec{0}$, $\beta_L = \beta_L'$ and $\beta_R = \beta_R'$. %
Let~$M$ be the sequence of steps in the window movie $M_{\vec{\alpha},w}$.
The algorithm in Figure \ref{fig:algo-seq} describes how to produce a new valid assembly sequence $\vec{\gamma}$.

\begin{figure}[t]
\begin{algorithm}[H]
\SetAlgoLined
Initialize $i$, $j$, $k = 0$ and $\vec{\gamma}$ to be empty

\While{$i<|\vec{\alpha}|$  { \bf or }  $j<|\vec{\beta}|$}{  %
  \If{$Pos(M[k]) \in \dom{\alpha_L}$}{ %
    \While{$i < |\vec{\alpha}|$ and $Pos(\vec{\alpha}[i])\neq Pos(M[k])$}{
      \If{$Pos(\vec{\alpha}[i]) \in \dom{\alpha_L}$}{$\vec{\gamma} = \vec{\gamma} + \vec{\alpha}[i]$}
      $i = i + 1$
    }
    \If{$i<|\vec{\alpha}|$}{
      $\vec{\gamma} = \vec{\gamma} + \vec{\alpha}[i]$

      $i = i + 1$
    }
  }
  \ElseIf{$Pos(M[k]) \in \dom{\beta_R}$}{
    \While{$j<|\vec{\beta}|$ and $Pos(\vec{\beta}[j])\neq Pos(M[k])$}{
      \If{$Pos(\vec{\beta}[j]) \in \dom{\beta_R}$}{
        $\vec{\gamma} = \vec{\gamma} + \vec{\beta}[j]$}

        $j = j + 1$

    }
    \If{$j<|\vec{\beta}|$}{
      $\vec{\gamma} = \vec{\gamma} + \vec{\beta}[j]$

      $j = j + 1$
    }
  }
  \ElseIf{$k\geq |M|$}{
    \If{$i<|\vec{\alpha}|$}{
      $\vec{\gamma} = \vec{\gamma} + \vec{\alpha}[i]$

      $i = i + 1$
    }
    \If{$j<|\vec{\beta}|$}{
      $\vec{\gamma} = \vec{\gamma} + \vec{\beta}[j]$

      $j = j + 1$
    }
  }

  $k = k + 1$
}
\Return $\vec{\gamma}$

\end{algorithm}
\caption{The algorithm to produce a valid assembly sequence $\vec{\gamma}$.}
\label{fig:algo-seq}
\end{figure}

If we assume that  the assembly sequence  $\vec{\gamma}$ ultimately produced by the algorithm is valid, then the result of $\vec{\gamma}$  is indeed  $\alpha_L \beta_R$, since for every tile in~$\alpha_L$ and $\beta_R$, the algorithm adds a step to the sequence $\vec{\gamma}$ involving the addition of this tile to the assembly. However, we need to prove that  the assembly sequence $\vec{\gamma}$  is valid,  it may be the case that either: 1. there is insufficient bond strength between the tile to be placed and the existing neighboring tiles, or 2. a tile is already present at this location.
Case 2 is a non-issue, as locations in $\alpha_L$ and $\beta_L$ only have tiles from $\alpha_L$ placed in them, and locations in $\alpha_R$ and $\beta_R$ only have tiles from $\beta_R$ placed in them.
Case 1 is more difficult, and is where the remainder of the proof is spent.

Formally, we claim the following: at each step of the algorithm, the current version of $\vec{\gamma}$ at this step is a valid assembly sequence whose result is a producible subassembly of $\alpha_L \beta_R$.
Note that the outer loop of the algorithm iterates through all steps of $\vec{\alpha}$ and $\vec{\beta}$, such that at any point of adding $\vec{\alpha}[i]$ (or $\vec{\beta}[j]$) to $\vec{\gamma}$, all steps of the window movie occurring before $\vec{\alpha}[i]$ ($\vec{\beta}[j]$) in $\vec{\alpha}$ ($\vec{\beta}$) have occurred.
Similarly, all tiles in $\alpha_L$ (or $\beta_R$) added to $\alpha$ ($\beta$) before step $i$ ($j$) in the assembly sequence have occurred.

So if the $Tile\left(\vec{\alpha}[i]\right)$ that is added to the subassembly of $\alpha$ produced after $i-1$ steps, can bond at a location in $\alpha_L$ to form a $\tau$-stable assembly, the same tile added to the producible assembly of  $\vec{\gamma}$  must also bond to the same location in  $\vec{\gamma}$, as the neighboring glues consist of (i) an identical set of glues from tiles in the subassembly of $\alpha_L$ and (ii) glues on the side of the window movie containing~$\alpha_R$.  Similarly, the tiles of $\beta_R$ must also be able to bind.

So the assembly sequence of $\vec{\gamma}$ is valid, i.e.\ every addition to $\vec{\gamma}$ adds a tile to the assembly  to form a new producible assembly.
Since we have a valid assembly sequence, as argued above, the finished producible assembly is~$\alpha_L \beta_R$.
\end{proof}

In the proof, we used the two identical window movies to ensure each step in the constructed assembly sequence was valid, i.e.\ the proposed tile could attach at the specified location.
However, if a pair of incident glues in the window movie are not identical, then they are never used to ensure a proposed tile can attach.
Using this observation, we define a restricted form of  window movie, called a \emph{bond-forming submovie}, which consists  of only those steps of the window movie that place glues that eventually form positive-strength bonds in the assembly.
Every window movie $M$ has a unique bond-forming submovie ${\cal B}(M)$, and Lemma~\ref{lem:windowmovie} can be strengthened by relaxing the requirement that the window movies $M_{\vec{\alpha}, w} = M_{\vec{\beta}, w'}$ match:

\begin{corollary}
\label{cor:windowmovie}
The statement of Lemma~\ref{lem:windowmovie} holds if the window movies $M_{\vec{\alpha},w}$ and $M_{\vec{\beta},w'}$ are replaced by  their bond-forming submovies ${\cal B}\left(M_{\vec{\alpha},w}\right)$ and ${\cal B}\left(M_{\vec{\beta},w'}\right)$.
\end{corollary}

\begin{proof}
The matching window movies $M_{\vec{\alpha},w}$ and $M_{\vec{\beta},w'}$ in the proof of Lemma~\ref{lem:windowmovie} are  used  only to prove that for each step (tile addition) of $\vec{\alpha}$ or $\vec{\beta}$ that is appended to the sequence $\vec{\gamma}$, the tile can attach at the new proposed location.
For each step of $M_{\vec{\alpha},w} = M_{\vec{\beta},w'}$, either the step is in ${\cal B}\left(M_{\vec{\alpha},w}\right) = {\cal B}\left(M_{\vec{\beta}, w'}\right)$ or not.
If so, the proof is unchanged.

Otherwise, if not, the tile will not form a bond with any glue (i.e.\ tile) on the other side of the window, since the step is not in ${\cal B}\left(M_{\vec{\alpha}, w}\right)$.
Furthermore, the set of glues incident to $Pos(\vec{\alpha}[i])$ (respectively, $Pos(\vec{\beta}[j])$) and forming positive strength bonds is identical to the set when $\vec{\alpha}[i]$ ($\vec{\beta}[j]$) is added to $\vec{\gamma}$ in the proof of Lemma~\ref{lem:windowmovie}, as all elements of $\vec{\alpha}$ ($\vec{\beta}$) preceeding $\vec{\alpha}[i]$ ($\vec{\beta}[j]$) have already been added to $\vec{\gamma}$.
\end{proof}

\subsection{The simulated tile set}
\label{sec:simulated tile set}

Here we describe the tile assembly system $\mathcal{T} = (T,\sigma,2)$ to be simulated by the claimed simulator tile set $U$.
The tile set $T$ consists of a small constant number of tile types as seen in Figure~\ref{fig:construction}:
the \emph{seed} $\sigma$,
eight \emph{arm} tiles,
six \emph{finger} tiles,
a \emph{keystone} tile,
a \emph{flagpole} tile,
and a \emph{flag} tile.
Of the infinite set of terminal assemblies formed, each assembly either contains both the keystone and flag tile types or does not (see Figure~\ref{fig:construction}).

\begin{figure}[ht!]
\centering
\includegraphics[scale=1.0]{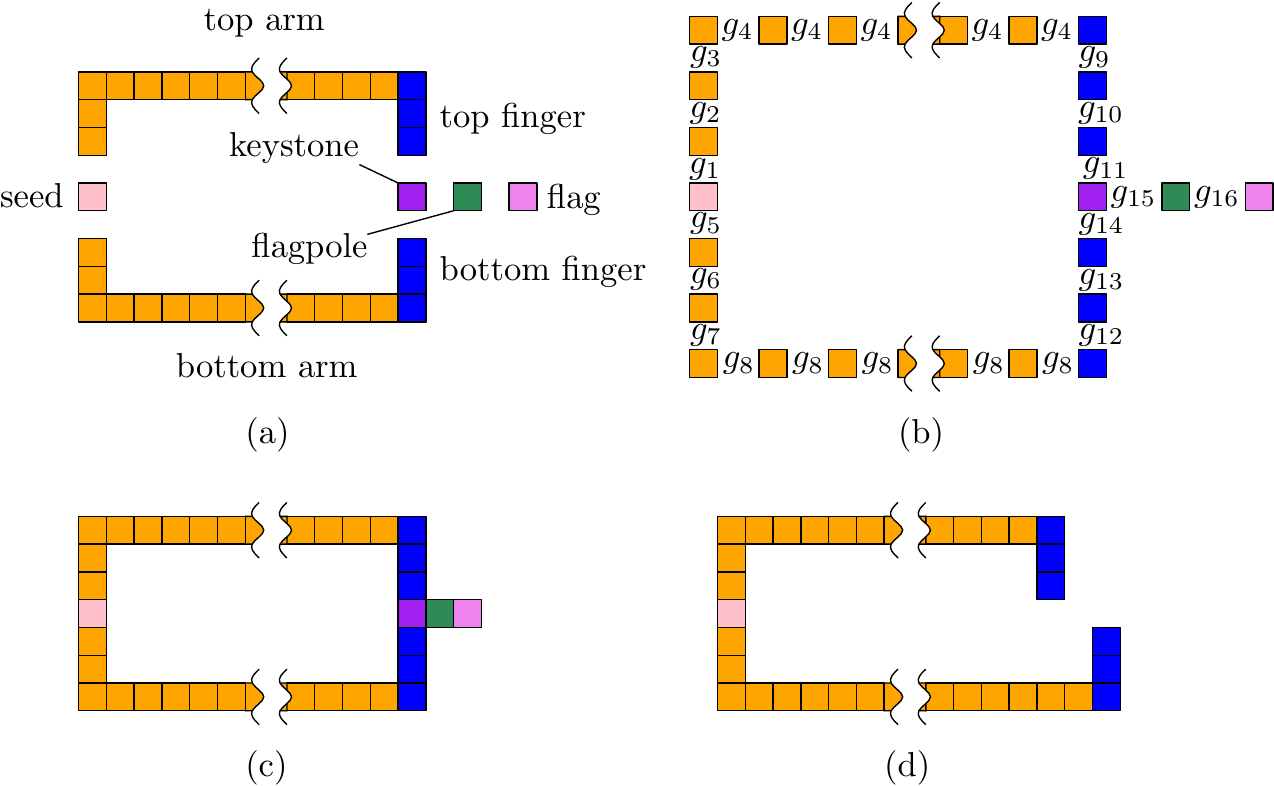}
\caption{(a) An overview of the tile assembly system $\mathcal{T} = (T,\sigma,2)$.~$\mathcal{T}$ runs at temperature 2 and its tile set $T$ consists of 18 tiles. (b) The glues used in the tileset $T$. Glues $g_{11}$ and $g_{14}$ are strength 1, all other glues are strength~2.  Thus the keystone tile binds with two ``cooperative'' strength~1 glues. Growth begins from the pink seed tile $\sigma$: the top and bottom arms are one tile wide and grow to arbitrary, nondeterministically chosen, lengths. Two blue figures grow as shown. (c) If the fingers happen to meet then the keystone, flagpole and flag tiles are placed, (d) if the fingers do not meet then growth terminates at the finger ``tips'':  the keystone, flagpole and flag tiles are not placed.}
\label{fig:construction}
\end{figure}

The glues in the various tiles are all unique with the exception of the common east-west glue type used within each arm to induce non-deterministic and independent arm lengths.
Glues are shown in part (b) of Figure~\ref{fig:construction}.
Note that cooperative binding happens at most once during growth, when attaching the keystone tile to two arms of identical length.
All other binding events are noncooperative and all glues are strength-2 except for $g_{11}, g_{14}$ which are strength-1.

Recall that a universal tile set $U$ simulating~$\mathcal{T}$ carries out the simulation by creating $m \times m$ supertiles that represent  the tiles of~$\mathcal{T}$, and that are placed with the same dynamics (i.e.\ tile placement ordering, modulo rescaling) as~$T$.
In particular, $U$ must simulate the creation of a terminal assembly with a flag by placing all of the supertiles in both arms first, then the keystone supertile, flagpole supertile, and finally flag supertile.
Though $U$ is permitted to place tiles in {\em fuzz} supertile regions (i.e.\ adjacent to supertile regions with a non-empty represented tile type), $U$ cannot put tiles in the flag supertile region before placing tiles that represent  the flagpole tile.
That is, any assembly sequence of $U$ placing a tile in the flag supertile region \emph{must} have already simulated an assembly sequence placing the flagpole tile, which in turn \emph{must} have already simulated an assembly sequence placing the keystone tile, and so on.

\subsection{Invalid simulation of $\mathcal{T}$}
\label{sec:invalid}

In this section we give the main proof argument for Theorem~\ref{thm:3d-not-iu} by showing that the  tile set $U$ does not simulate~$\mathcal{T}$.

Let $g$ be the number of glues in the tile set $U$ and let~$m$ be the scale factor chosen for~$\mathcal{T}$.
For the remainder of the proof, we only consider the simulation by $U$ of~$\mathcal{T}$ in the case that~$\mathcal{T}$ grows an assembly $\gamma$ with a pair of arms of identical horizontal length $((g+1)^{6m} \cdot (6m)! + 1) \cdot 3 + 6$.
This length is justified as follows.

By Definition~\ref{def-s-models-t}, there exists $\gamma' \in \mathcal{A}[\mathcal{U}]$ such that $R^*\left(\gamma'\right) = \gamma$, where~$\mathcal{U} = (U, \sigma_{\mathcal{T}}, 1)$ is the simulator tile assembly system using tile set $U$, seed assembly~$\sigma_{\mathcal{T}}$, and  temperature~1. The simulator uses scale $m$, therefore because the definition of \emph{cleanly maps to} (see Section~\ref{sec:simulation_def}) permit one-supertile wide ``fuzz'' (i.e. the placement of tiles in locations adjacent to supertiles but which don't map to a tile in $\mathcal{T}$), the vertical height of an arm  is at most $3m$.
Any window that cuts the bottom arm of the simulation $\alpha'$ vertically, has one of $(g+1)^{6m}$ sets of glues corresponding to $6m$ locations that glues can appear at and the $g+1$ distinct choices for each glue (including the null glue).
So any window movie that vertically cuts the bottom arm of the assembly has such a glue set, and one of at most $(6m)!$ possible orderings for these glues to appear in the movie.
Then by the pigeonhole principle, examining $((g+1)^{6m} \cdot (6m)! + 1)$ such vertical cuts ensures some set of $6m$ glues and their ordering occurs twice.
If the arm has length $((g+1)^{6m} \cdot (6m)! + 1) \cdot 3 + 6$, examining one vertical cut of the bottom arm in every third supertile of the simulation, ignoring the first and last three supertiles in the arm, also finds a set of $6m$ glues and their ordering that occurs twice.

We now show how to combine this fact with Corollary~\ref{cor:windowmovie} to construct an assembly, producible by the simulator, but  that is not a simulation of any assembly produced by $\mathcal{T}$. Let $\vec{\gamma}' = \left(\gamma'_i \mid 0 \leq i < k\right)$ be such that the result of $\vec{\gamma}' = \gamma'$. Consider the first step $i$ of the assembly sequence $\vec{\gamma}'$ that places a tile~$t$ at some location $\vec{x}$, i.e., $\gamma'_{i} = \gamma'_{i-1} + \left(\vec{x} \mapsto t\right)$, satisfying one of the following two conditions:

\begin{enumerate}
\item[(1)] the placement of tile $t$ completes a path between the top finger and bottom finger through the keystone supertile, and possibly also through the $m \times m$ region of fuzz immediately to the west of the keystone supertile;
\item[(2)] the placement of tile $t$ is in the flagpole supertile.
\end{enumerate}

\newcommand\assemblyatstep[1]{C_{#1}}

\begin{figure}[t] %
\centering
\includegraphics[scale=1.0]{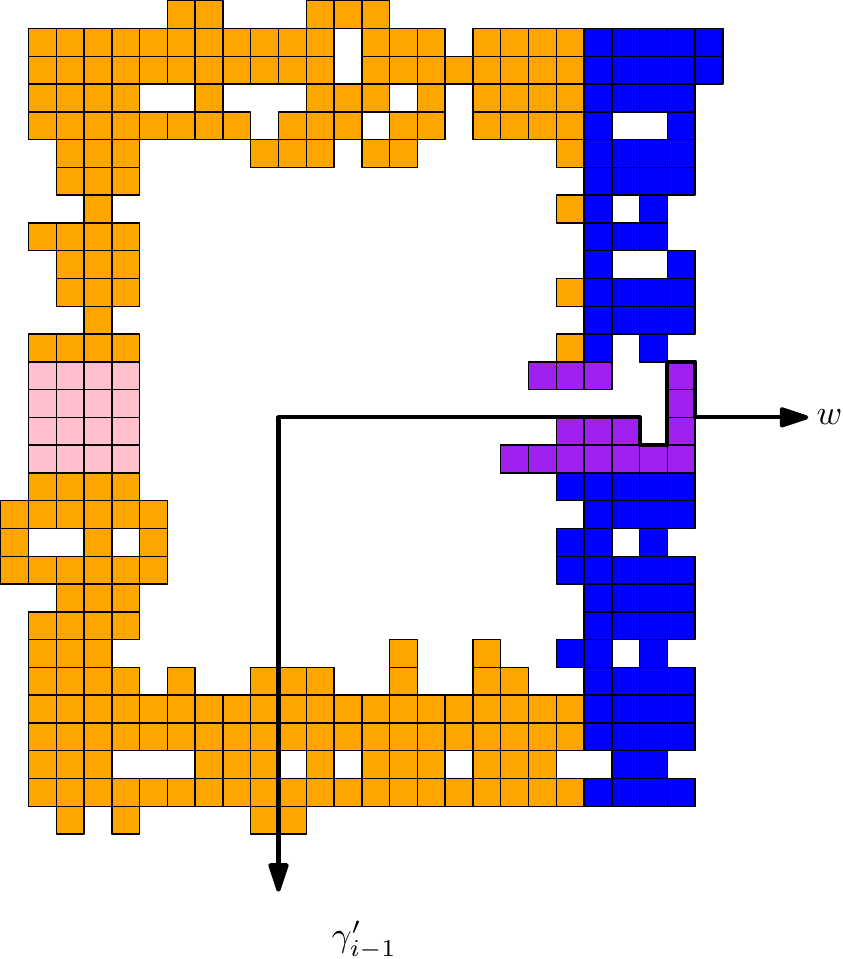}
\caption{The assembly $\gamma'_{i-1}$ and a window $w$ formed from (i) a vertical cut of the bottom arm and (ii) a path through the keystone region that does not cross any bond in the keystone region nor in the fuzz region to the west of the keystone region.
The bond-forming submovie ${\cal B}\left(M_{\gamma'_{i-1},w}\right)$ has no glues in the keystone region of $\gamma'_{i-1}$, since no path in $\gamma'_{i-1}$ from the top finger to the bottom finger through the keystone region exists.}
\label{fig:finger-window}
\end{figure}

Now, step backwards in the assembly process by one step and consider $\gamma'_{i-1}$, i.e., the assembly at step ${i-1}$ of $\vec{\gamma}'$.
Since condition (1) has not occurred, there exists a path $p$ along the edges of the grid graph starting from the $m\times m$ region that is distance $2m$ west from the keystone supertile, which travels eastward, threading through the $m\times m$ region west of the keystone supertile, then continues threading through the keystone supertile, and then past the east extent of $\gamma'_{i-1}$, such that no edge of $p$ crosses an edge shared by matching glues in $\gamma'_{i-1}$ (see Figure~\ref{fig:finger-window}).
So for any vertical cut of the bottom arm of $\gamma'_{i-1}$, one can extend the vertical cut into a window such that the bond-forming submovie of the window only has glues in the vertical cut of the bottom arm of $\gamma'_{i-1}$ (again, see Figure~\ref{fig:finger-window}).

\begin{figure}[t]
\centering
\includegraphics[width=1.0\columnwidth]{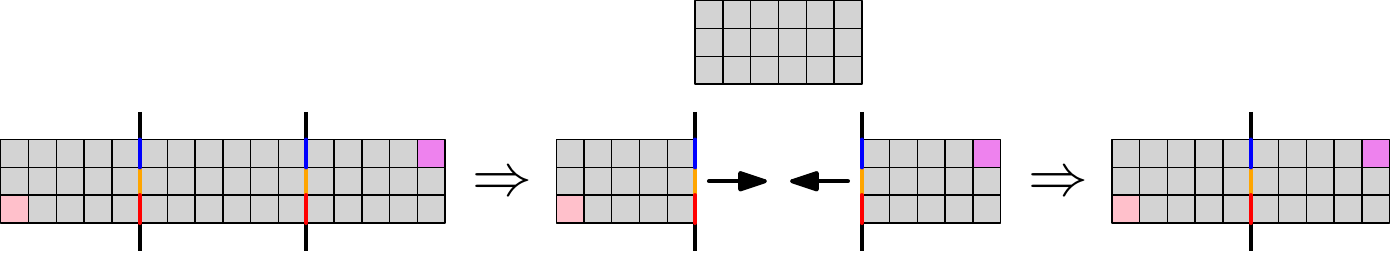}
\caption{Splicing two identical window movies together to produce another valid assembly sequence and terminal assembly by invoking Corollary~\ref{cor:windowmovie} with vertical windows.}
\label{fig:splicing}
\end{figure}

Then by the previous counting argument, one can find two such windows $w$, $w'$ with identical bond-forming submovies, as these windows only have glues forming bonds in the vertical cut of the bottom arm (see the left part of Figure~\ref{fig:splicing}).
Moreover, the two windows have vertical cuts separated horizontally by distance $d \geq 3m$ and not occurring in the first or last three supertiles of the arm.

This last property is key, as it follows $w$ and $w'$ can be modified to follow the same path through the keystone supertile or the fuzz immediately west by selecting a path through these supertiles and duplicating this path twice on both $w$ and $w'$.
The two occurrences of this subpath should be separated horizontally by distance $d$.
Then $w' = w + ( d, 0 )$.

\begin{figure}[t]%
\centering
\includegraphics[scale=1.0]{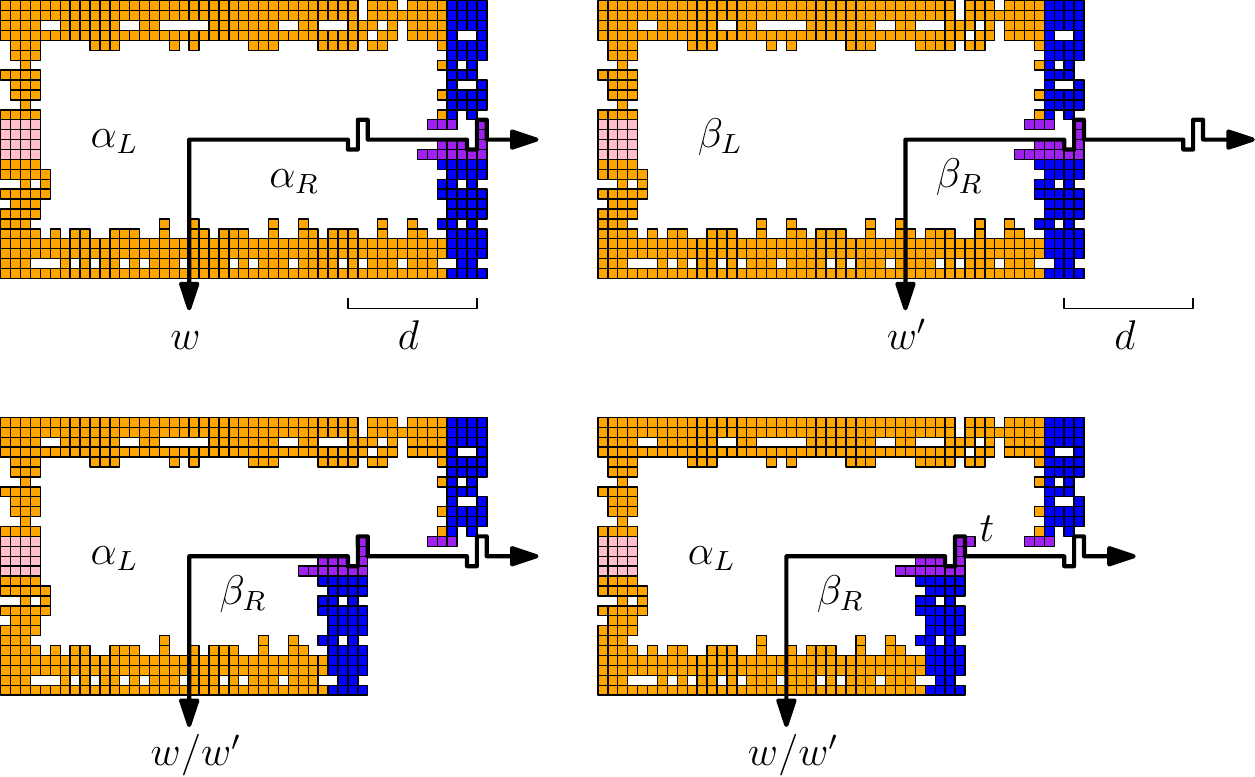}
\caption{An example of an assembly formed by $U$ simulating $\mathcal{T}$ and the identical bond-forming submoviews $w$ and $w'$ (top), the resulting producible assembly constructed via Corollary~\ref{cor:windowmovie} (bottom left), and producing an invalid simulation assembly by the valid placement of  a single tile $t$ (bottom right).}
\label{fig:bad-sim-overview}
\end{figure}

At this point, we have two assemblies $\alpha = \gamma'_{i-1}$, $\beta = \gamma'_{i-1}$, with assembly sequences $\vec{\alpha} = \vec{\beta} = \left(\gamma'_0, \ldots, \gamma'_{i-1}\right)$ and two identical bond-forming submovies ${\cal B}\left(M_{\vec{\alpha}, w}\right)$, ${\cal B}\left(M_{\vec{\beta}, w'}\right)$ for the assembly sequences of $\vec{\alpha}$ and $\vec{\beta}$ (see Figure~\ref{fig:bad-sim-overview}).
Then by Corollary~\ref{cor:windowmovie}, the assembly formed by taking the union of the assemblies consisting of 1. the part of $\gamma'_{i-1}$ partitioned by $w$ and containing the seed ($\alpha_L$), and 2. the part of $\gamma'_{i-1}$ partitioned by $w'$ and not containing the seed ($\beta_R$), denoted as $\alpha_L\beta_R$, is also a producible assembly of the simulation, i.e., $\alpha_L\beta_R \in \mathcal{A}[\mathcal{U}]$.
This assembly has a top arm of length $((g+1)^{6m} \cdot (6m)! + 1) \cdot 3 + 6$ supertiles and a bottom arm of length at least $6$ and at most $((g+1)^{6m} \cdot (6m)! + 1) \cdot 3 + 3$ supertiles.

Finally, we use information about which condition occurs in step $i$ of the simulation to construct an invalid assembly.
From conditions (1) and (2) above we know that $t$ binds to one of $\alpha_L$ or $\beta_R$.
Let $\hat{\gamma} = \alpha_L\beta_R + \left( \left(\vec{x}-(d,0)\right) \mapsto t \right)$, i.e., the addition of $t$ to $\alpha_L\beta_R$ at the relevant location.

If condition (2) holds (flagpole), then $t$ is placed in a region in which no tile should exist in a simulation with arms not aligned (fuzz in this region is not permitted, by the definition of (diagonal) fuzz in  Section~\ref{sec:simulation_def}). If condition~(1) holds, then $t$ was originally placed to complete a path between the tips of the top and bottom fingers through the keystone region in $\gamma'_i$.
So from $\hat{\gamma}$ we continue placing tiles found on the portion of this path from $t$ to the (here, nonexistent) top or bottom finger, so that we are recreating exactly the path between finger tips found in~$\gamma_i$.
Note that these new tiles are all placed within the keystone region and the supertile immediately to the west of the keystone, with the exception of exactly one tile placed either in the $m \times 1$ row of tile locations directly above the keystone above the (shorter) bottom arm (if $t$ was bound to $\beta_R$), or in the $m \times 1$ row of tile locations directly below the keystone below the (longer) top arm (if $t$ was bound to $\beta_R$).
In either case, Definition~\ref{def-t-follows-s} says that the placement of this particular tile implies an invalid simulation by a producible assembly.

To conclude the proof of Theorem~\ref{thm:3d-not-iu}, any claimed universal tile set~$U$ simulating~$\mathcal{T}$ (and in particular the assembly processes with very long but equal-length arms) produces assemblies that do not  correspond to a simulation of any assembly produced by $\mathcal{T}$.
That is, $U$ does not correctly simulate the production of $\mathcal{T}$ (Definition~\ref{def-equiv-prod}), hence $U$ does not correctly simulate~$\mathcal{T}$ (Definition~\ref{def-s-simulates-t}), and since nothing was assumed about~$U$ other than its existance, no such universal tile set  exists. \qed

\section{3D temperature-1 aTAM simulates 2D \\ temperature-1 aTAM}
\label{sec:3D_temp1_simulates_2D_temp1}

In this section, we give a proof of Theorem~\ref{thm:3D-sim-2D-simple-statement}. Formally, we  show that there exists a 3D aTAM tile set $\mathcal{U}$ such that, given an arbitrary 2D aTAM tile system $\mathcal{T} = (T,\sigma,1)$, where $|\sigma|=1$, there exists an appropriately initialized seed assembly $\sigma_T$, which depends on $T$, such that $\mathcal{U} = (U,\sigma_T,1)$ simulates $\mathcal{T}$ at scale factor $c$, for some $c \in \mathbb{N}$.

Our construction makes use of several of the techniques from \cite{USA}.  The basic idea is to use the tiles of $\mathcal{U}$ to assemble three-dimensional volumes, called \emph{supertiles}, each of which represent a single tile from $T$.  The dimensions of each supertile are $c \times c \times 6$.  The initial supertile which represents $\sigma$ contains an encoding of the entire tile set $T$. This encoding is ``passed'' from each supertile to each newly forming supertile and is used by each supertile to determine the tile type of $T$ that the supertile is supposed to simulate. The encoding of $T$ is also used by each supertile to determine any ``output'' glues, which may contribute to, if not initiate, the growth of neighboring supertiles.

Before presenting our construction, we first define a useful self-assembly gadget for reading geometrically-specified input.

\subsection{Read-write gadgets}

In temperature $2$ tile assembly systems, a tile attachment can be the result of the binding of two strength $1$ glues on different sides of the tile.  We call this \emph{cooperative} binding since the two tiles to which the new tile is binding are ``cooperating'' to allow for its attachment by each sharing a glue, and thus the information encoded in that glue.  However, in temperature $1$ systems such behavior cannot be enforced because either glue of the pair is sufficient to allow a new tile to bind, possibly ignoring the second glue.  This means that if, in order for the correct tile to be placed, it must ``collect'' information from more than one adjacent tile, then this information cannot be transmitted strictly via glues interacting.

One solution to this problem, in 3D, is to grow a path of tiles, which can potentially split into two (or more) branches, and use a previously placed tile to block the growth of one branch but allow further growth of another branch.  The path allowed to continue is thus explicitly provided with information from the glues along the path, as well as implicitly from the fact that it gets to continue. This is a  method to handle the fact that we can not do  cooperative binding at temperature 1: it uses geometry to transmit the ``second'' piece of information that must be used to make a decision.  In order to ensure that such information is deterministically provided to the growing path, the tiles which block one branch from completing must be guaranteed to have been placed prior to the growth of the path. Since it is possible for any branching paths at temperature $1$ to grow independently of each other, with either branch growing arbitrarily far before the other is extended by even a single tile, it is necessary to force the growth of portions of the assembly, which require such behavior, to be restricted to be a single-tile-wide path. Such a path will zig-zag back and forth in order to ``read'' the information previously encoded in the  geometric placement of blocking tiles.

See Figures~\ref{fig:3D-gadget-arm-write0}-\ref{fig:3D-arm} for examples of a path encoding each of two possible values, which are read by a later portion of the same path.

\begin{figure}
\begin{center}
\includegraphics[width=4.5in]{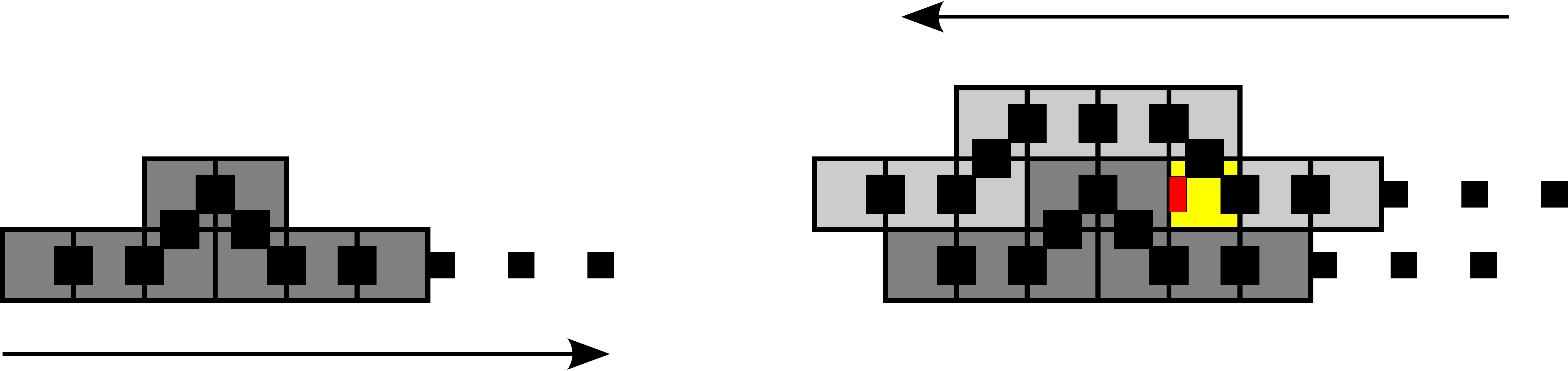}
\caption{An example of a (2D) path of tiles (left) growing to the right and encoding a `$0$', and (right) growing back to the left and reading the `$0$'. The reading path can potentially branch at the location denoted by the yellow tile.  However, one possible branch is blocked, with the mismatched and blocked glue shown in red.}
\label{fig:3D-gadget-arm-write0}
\end{center}
\end{figure}

\begin{figure}[t]
\begin{center}
\includegraphics[width=4.5in]{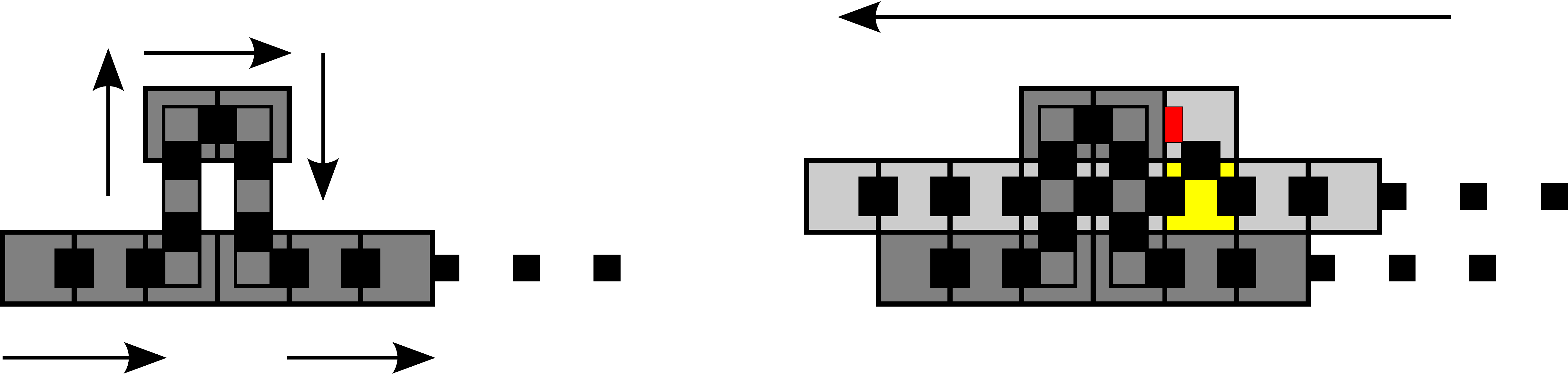}
\caption{An example (in 3D) of a path of tiles (left) growing to the right and encoding a `$1$', and (right) growing back to the left and reading the `$1$'. See Figure~\ref{fig:3D-arm} for a 3D view of the portion of the path encoding the `$1$'.  Note that the smaller grey squares denote tiles which are located in the $z=1$ plane, while the others are located in $z=0$.  See also Figure~\ref{fig:3D-gadget-arm-write0} for more explanation.}
\label{fig:3D-gadget-arm-write1}
\end{center}
\end{figure}

\begin{figure}[t]
\centering
    \subfloat[3D view of path encoding value `$1$']{
        \label{fig:3D-arm1}
        \includegraphics[width=1.5in]{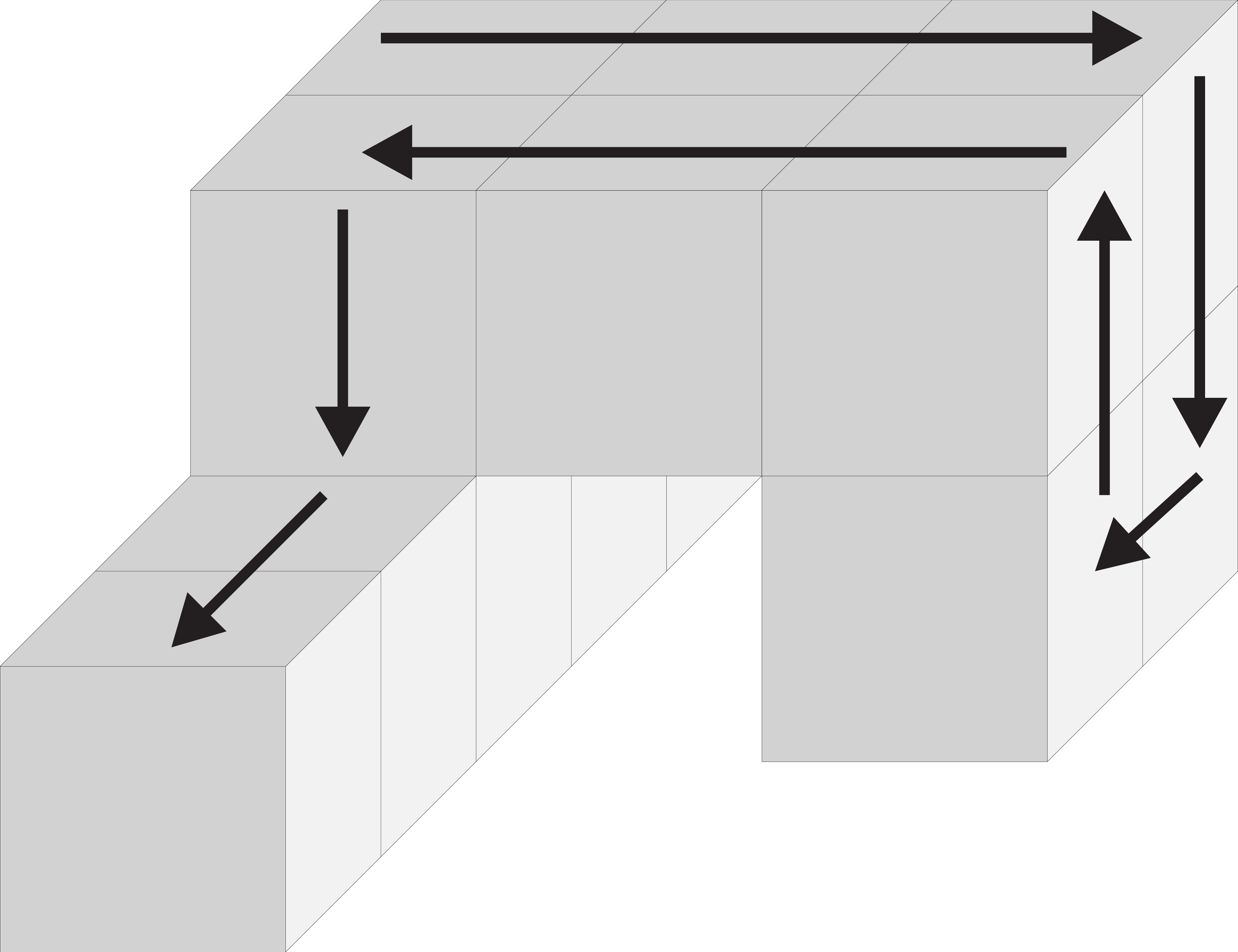}}
    \quad\quad\quad
    \subfloat[Rotated 3D view of path encoding value `$1$']{
        \label{fig:3D-arm2}
        \includegraphics[width=2.2in]{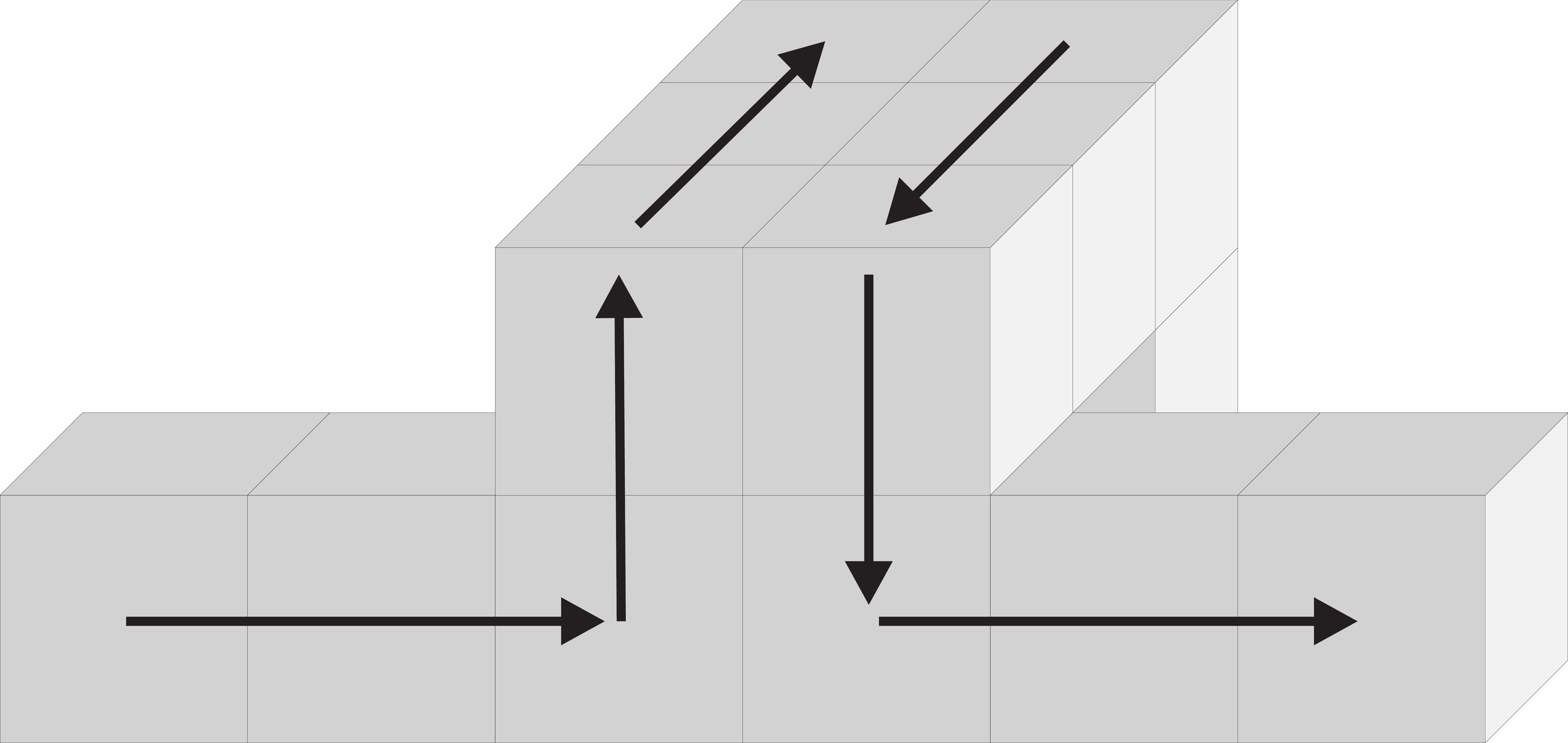} }
    \caption{\small 3D view of a path encoding a value by stepping up into the third dimension then across and back down to place a blocking tile, then growing back to continue along the original direction.  The arrows show the order of tile additions.  See Figure~\ref{fig:3D-gadget-arm-write1} for more details on the growth of the path.}
\label{fig:3D-arm}
\end{figure}

While the examples of Figures~\ref{fig:3D-gadget-arm-write0}-\ref{fig:3D-arm} demonstrate the ability of a short path of tiles to read one of two possible values (a single bit), we can combine these gadgets to read a sequence of values.

For example, suppose we encode an input string $w = w_0\cdots w_{l-1}$ as a series of geometric ``bumps'' and ``dents'' along a path, then the glues (that connect the tiles) of a path that ultimately navigates these geometric obstacles can effectively read each bit $w_i$, such that after the ``reader'' path finishes scanning all of the bumps and dents, the input $w$ is stored in the most-recently-added tile $t$ in the reader path. Then it is possible to use $t$ to compute some function $z = f(w)$ of those bits. Then the growth of a final ``output'' path can be initiated which goes go and builds a path representing the correct pattern of bumps and dents corresponding to the value $z$. These output bumps and dents can then be used as input for a subsequent ``reader'' path.

In order to modularize such functionality, we now define a \emph{read-write gadget}, which is a block of depth 2 or 4 that can be used in any of the three layers $L_0$, $L_1$ and $L_2$ by simply translating it to the planes $z=0$, $z=2$ or $z=4$, respectively.  A read-write gadget is a $g \times g \times d$ region, for some $g \in \mathbb{N}$ and $d \in \{2,4\}$, with (1) an \emph{entrance} location, (2) a \emph{reading} region, (3) at least one \emph{output} region and (4) zero or more \emph{exit} locations. In the reading region, a path implicitly reads the geometry of a previously-assembled path of tiles via a series of branching points at which the path may branch one of two possible ways depending on a bit value specified geometrically. An output region is where the path travels after it has finished collecting the input bits specified by the geometry of the reading region, and a single read-write gadget may have output regions on up to 2 different planes (i.e. 0 and 2, or 2 and 4), which is the reason that they may be of depth either 2 or 4, and this is the way that we will transfer information among different levels of the construction.  An output region of one read-write gadget may overlap with neighboring read-write gadgets to so that the output of one read-write gadget can serve as the input for the reading section of another read-write gadget. An exit location is where the a path exits the gadget. Note that, after read-write gadget completes its reading phase, its reading path may branch into multiple output paths, whence a read-write gadget may have more than one exit location.  See Figure~\ref{fig:3D-sim-2D-gadget-example} for an example of a read-write gadget, which reads a series of bits $A$, $B$, and $C$ (specifically, $A=0$, $B=1$, and $C=0$), and then outputs the bits ($A=1$, $B=0$, $C=0$) before exiting.  Note that the input and output are both located on plane $z=0$, while both require the placement of some tiles into plane $z=1$.

\begin{figure}
\begin{center}
\includegraphics[width=5.0in]{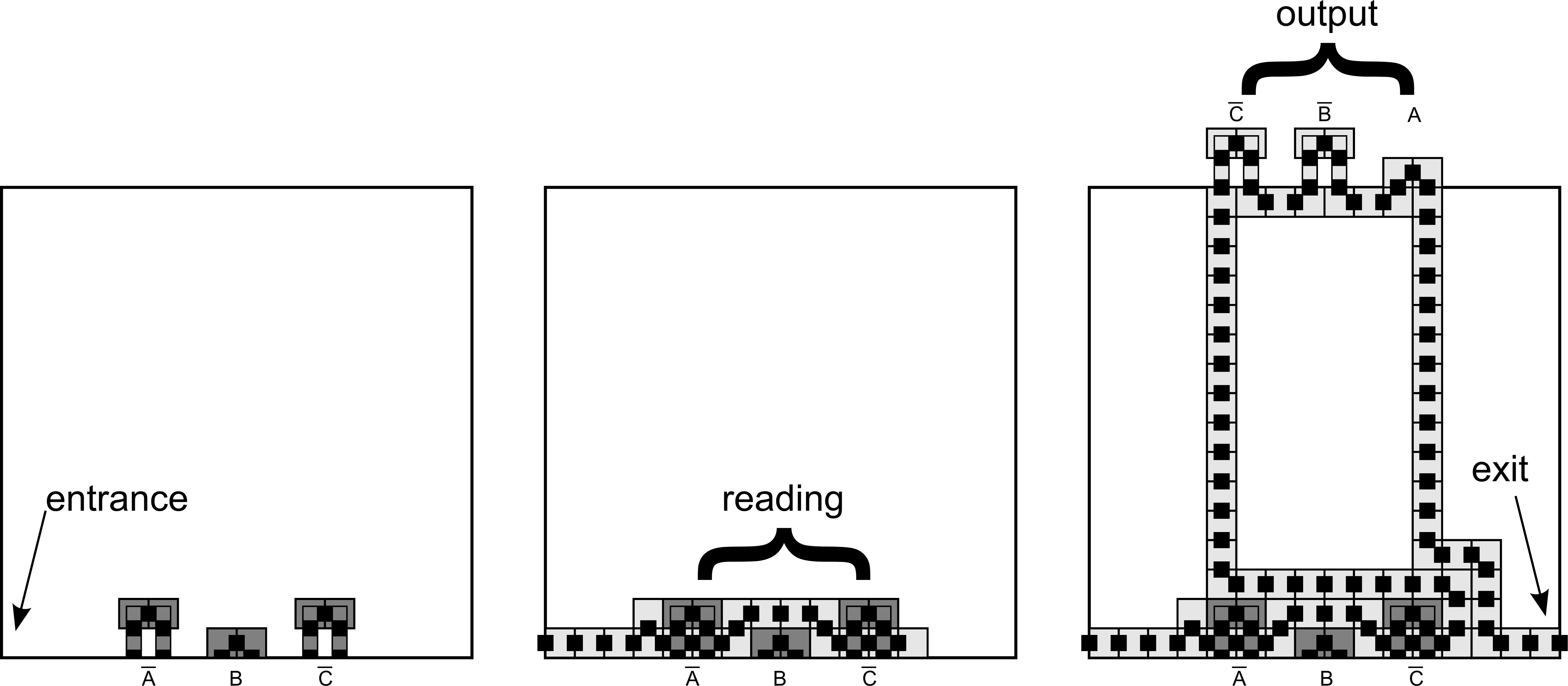}
\caption{An example of a read-write gadget which reads the bits $A=0$, $B=1$, and $C=0$ and outputs $A=1$, $B=0$, $C=0$.  The tiles attached to those providing the input but outside of the gadget are not shown.  The gadget requires the use of planes $z=0$ and $z=1$, but both input and output are located in plane $z=0$.}
\label{fig:3D-sim-2D-gadget-example}
\end{center}
\end{figure}

Given a constant sized tile set, the number of input values which can be read within a read-write gadget is bounded by a constant, since we propagate information about each branch solely by the glues along the path. In other words, after a reading path navigates each geometrically-specified bit, that bit is concatenated to each of the subsequent glues along the path. Interestingly, this idea can be carried out in very much the same spirit at temperature 1 in 2D if one negative-strength glue is allowed (i.e., a glue that potentially subtracts from--instead of adds to--the total strength with which a tile may bind) \cite{Patitz-2011}.

In our construction, we will make the following simplifying assumptions: 1) all read-write gadgets have input on no more than two sides (with respect to the $x$ and $y$-axes) and located in a single plane, 2) all reading and output regions are on planes $z=0,2,$ or $4$, with planes $z=1,3,$ and $5$ reserved for the paths needed to ``reach over'' and construct output regions, 3. entrance and exit regions are never on the same side and plane, and 4) reading and output regions are never on the same side and plane. 

\subsection{Construction details}

First, we divide the 3D space into \emph{layers}, each of which consist of two consecutive planes.  We define layer $L_0$ as planes $z=0,1$, $L_1$ as planes $z=2,3$, and $L_2$ as planes $z=4,5$.  %
We call $L_0$ the \emph{competition layer}, and it is used to ``decide'' which input superside is responsible for choosing the tile from $T$ to be represented and for creating the output supersides.  We call $L_1$ the \emph{information layer}, and it is used to help propagate the information to and from input and output supersides. $L_2$ is the \emph{output layer} and it selects and distributes the necessary output information for each superside, in effect arranging the ``output'' glues for each simulated tile (as well as the full definition of $T$) into the proper locations to serve as inputs for subsequent supertile formation.

Note that $6$ planes in $z$ are not strictly necessary for this construction, and although it can be made to work in $3$ (or perhaps even a minimum of $2$), modifying the construction to use fewer than $6$ planes makes it more complicated and more difficult than it already is to present: therefore, we choose $6$ for clarity of presentation.  In an effort to simplify the construction for presentation, we describe it in such a way that we subdivide each supertile into a grid of read-write gadgets (all of the same dimensions and with read and output locations for the same set of variables) rather than individual tiles. This will come at the cost of a larger overall scale factor for the simulation, but only by a constant independent of the tile set being simulated.

To help describe our construction, we make use of an example throughout.  The tile set used for the example can be seen in Figure~\ref{fig:3D-sim-2D-example-tile-set}.  The first aspect of the construction which we will explain is the encoding of $\mathcal{T}$ by the tiles of $\mathcal{U}$.

\begin{figure}
\begin{center}
\includegraphics[width=1.5in]{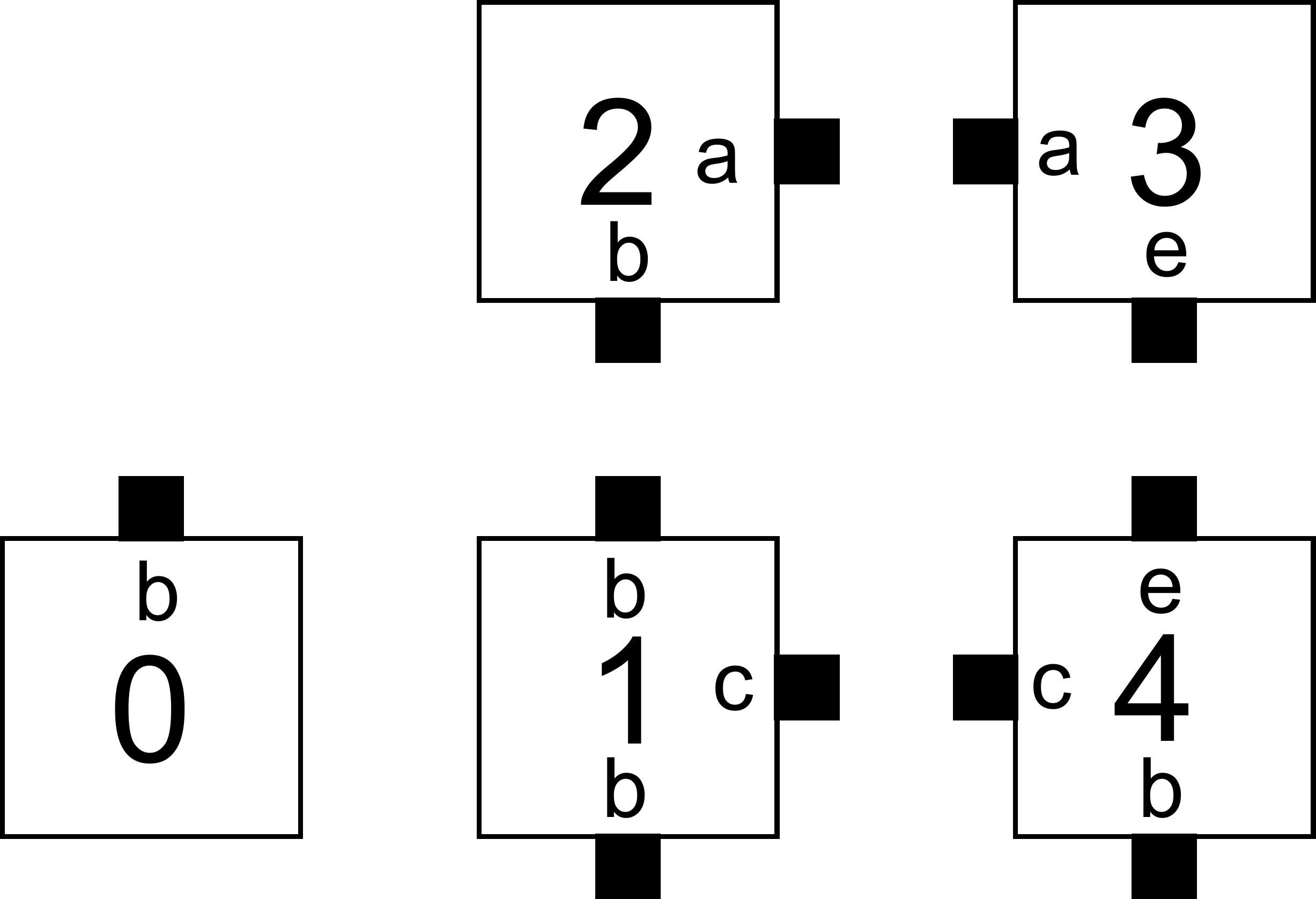}
\caption{An example tile set used to describe the construction for Theorem~\ref{thm:3D-sim-2D-simple-statement}.}
\label{fig:3D-sim-2D-example-tile-set}
\end{center}
\end{figure}

\subsubsection{Encoding $T$}\label{sec:3D-sim-2D-T-encoding}

To encode $T$, we make use of the fact that, at temperature $1$, the glue on the edge of a tile characterizes the set of all tiles capable of binding to that edge of the tile in any producible assembly (this stands in contrast to temperature $2$ systems in which a single strength $1$ glue only specifies half of the information for a potential binding event).  Therefore, rather than encode any information about the specific glues in $T$, we simply keep track of all tiles with the ability to bind to each side of each given tile.  To do so, let $t_0, t_1, \ldots, t_{|T|-1}$ be an enumeration of the tile types in $T$. We then place read-write gadgets composed of tiles in $\mathcal{U}$ in a line so that the edges on a given side output the  pattern defined by the algorithm in Figure~\ref{fig:encodingalgo}.

\begin{figure}[t]
\begin{algorithm}[H]
  Print $`B'$ %

  \For{$0 \leq i < |T|$}{
    \tcc{print $t_i$ as follows:}
    Print the binary number $i$, padded with 0's to length $\lceil \log |T| \rceil + 1$
    \For{$d \in \{N,E,S,W\}$}{
      \tcc{print side $d$ of tile $t_i$ as follows:}
      \tcc{print the special character `$N$',`$E$',`$S$', or `$W$' corresponding to the side}
      Print $d$

      \For{$0 \leq j < |T|$}{
        \If{side $d$ of $t_i$ binds to the opposite side of $t_j$}{
          $last \gets True$

          \For{$j < k < |T|$}{
            \tcc{determine if this is the final tile that binds}
            \If{side $d$ of $t_i$ binds to the opposite side of $t_k$}{
              $last \gets False$
            }
          }
          \eIf{$last == True$}{
            Print `$f$'
          }{
            Print `$y$'
          }
        }
        Print $`n'$%
      }
      Print the binary number $j$ padded to length $\lceil \log |T| \rceil$
    }
    Print $`D'$%
  }
  Print $`F'$%
\end{algorithm}
\caption{Algorithm that describes an encoding of tiles as strings.}
\label{fig:encodingalgo}
\end{figure}

This encoding is simply a listing of each tile type $t \in T$ which includes, for each tile and each direction, a full list including the number of each tile type $t' \in T$ and a `$y$' if $t$ and $t'$ bind along that edge of $t$ (i.e.\ their glues match)  and a `$n$' if they don't bind.  Further, if $t'$ is the last tile type in the list which does bind, instead of a `$y$', it is prefaced with an `$f$'.  (Note that the `$y$', `$f$', or `$n$' come before the number of the tile type in the enumeration of $T$.)  Thus, the encoding (with numbers written in decimal rather than binary and spaces added to make it easier to read) of the example $T$ from Figure~\ref{fig:3D-sim-2D-example-tile-set} is as follows:

B 0 Nn0y1y2n3f4 En0n1n2n3n4 Sn0n1n2n3n4 Wn0n1n2n3n4 D

  1 Nn0y1y2n3f4 En0n1n2n3f4 Sy0f1n2n3n4 Wn0n1n2n3n4 D

  2 Nn0n1n2n3n4 En0n1n2f3n4 Sn0n1n2n3f4 Wn0n1n2n3n4 D

  3 Nn0n1n2n3n4 En0n1n2n3n4 Sn0n1n2n3f4 Wn0n1f2n3n4 D

  4 Nn0n1n2f3n4 En0n1n2n3n4 Sy0f1n2n3n4 Wn0f1n2n3n4 D F

\subsubsection{Supersides}

We define a \emph{superside} to be the outermost row of read-write gadgets along the perimeter of one side of a $c \times c$ supertile in the construction.

Every supertile, other than the seed (see Section~\ref{sec:3D-sim-2D-seed-structure} for the structure of the seed) grows from an \emph{input} superside, which grows from the adjacent \emph{output} superside of a neighboring supertile.  An input superside for a supertile consists of the following components:

\begin{enumerate}
\item A binary string $h$, which encodes the height of the \emph{probe} (to be defined later),

\item A list $S$ containing the number of each tile type that could
  bind to the output superside, which placed this input superside
  (i.e.\ adjacent to this superside), and

\item The encoding of $T$ (previously discussed).
\end{enumerate}

The list $S$ is simply the list of tile types with each preceded by a `$y$', `$f$', or `$n$', corresponding to whether or not the supertile could grow to represent a tile of that type.  (Yes for those preceded with `$y$' or `$f$', no otherwise, with `$f$' marking the last one.)  See Section~\ref{sec:3D-sim-2D-T-encoding} for more detail. For example, an input superside for a supertile north of a supertile representing tile type~$1$ would have $S$ encoded as follows (with the numbers represented in binary): ``Nn0y1y2n3f4'', thus denoting that a tile of type $1,2,$ or $4$ could bind to the north of a 1 tile.

See Figure~\ref{fig:3D-sim-2D-competition-level} for an example of a supertile with all $4$ input supersides represented (along with the probes reaching toward the center of the supertile to be described in Section~\ref{sec:3D-sim-2D-competition-layer}).

\subsubsection{The competition layer}\label{sec:3D-sim-2D-competition-layer}

\begin{figure}
\begin{center}
\includegraphics[width=4.0in]{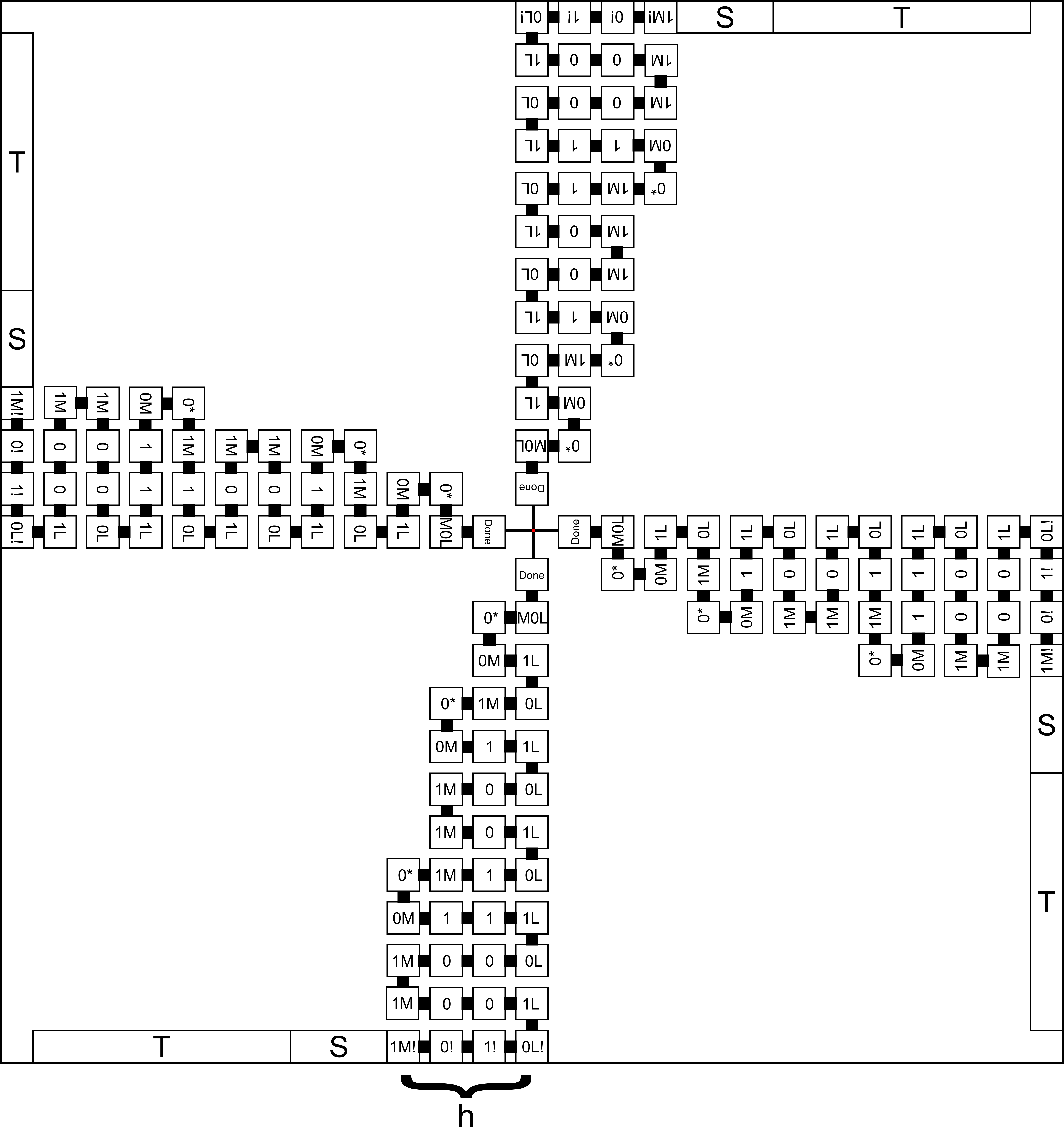}
\caption{A high-level depiction of the competition layer (not to scale).  The white squares logically represent tiles as read-write gadgets (rather than individual tiles), and the black squares are used to show the path of glue connectivity between consecutive gadgets and thus the zig-zag growth pattern.  Note that, although there is no glue binding between most gadgets in adjacent rows, the information is passed from a previous row to a subsequent row of read-write gadgets via the geometry of the output regions from the read-write gadgets in the previous row.}
\label{fig:3D-sim-2D-competition-level}
\end{center}
\end{figure}
The competition layer, $L_0$, is the arena in which a battle ensues (between competing probes) to determine the type of tile to be simulated by the newly-forming supertile. Assume that one or more input supersides for a supertile have formed (the seed supertile will have at least one output superside to be used as an input superside for supertile that represents a tile capable of binding to the seed in the simulated system).  Each such superside will begin the growth of a log-width binary counter  that counts down, beginning from the value $h$, encoded in the region dednoted by $h$ to $0$.  This pattern of growth is called a \emph{probe}, and grows to the location immediately adjacent to the center location of the supertile (the reader should consult the references \cite{IUSA,USA} for 2D simulation constructions implementing probes as decreasing binary counters). The center location of a supertile is not formed as a read-write gadget, but instead each probe attempts to grow a single-tile-wide path of tiles from the adjacent read-write gadget to place a tile in the center of the supertile.  Exactly one probe will win the competition to reach that center location first and be able to place a tile in that center position, thus ``winning the competition'' to determine what type the supertile will be. Note that this ``competition'' does not determine which tile type is to be simulated by this supertile. At this point, we only know the new supertile will grow from the winning superside.
After its victory, the supertile will then be able to form the output supersides. Growth of all other (losing) probes is halted by them being blocked from the winning (center) position. Thus, losing probes never leave the competition layer.

\def\hyph{-\penalty0\hskip0pt\relax}

Note that as a probe grows, all of the read-write gadgets along its counter-clockwise-most side, other than at the very base of the probe, present a special marker value. The read-write gadget closest to the input superside presents another special marker value denoting the end of the probe. The latter marker will be used by the path growing back along the counterclockwise-most side of the winning probe from the (winning) center to the superside from which the winning probe originated as a halting signal.  See Figure~\ref{fig:3D-sim-2D-probe} for an example of a probe winning the competition and then growing a path back down to the superside from which it originated.

\begin{figure}
\centering
    \subfloat[Probe growing upward as a log-width binary counter (that counts down, not up), where each labelled white square is a read-write gadget.  After counting down all the way to $0$, a single row of tiles (as opposed to a read-write gadget) grows toward the center location.  If it is able to place a tile there, the probe ``wins'' the competition thus determining the identity of the newly-forming supertile.]{
        \label{fig:3D-sim-2D-probe-win}
       \hspace{4ex} \includegraphics[height=5.0in]{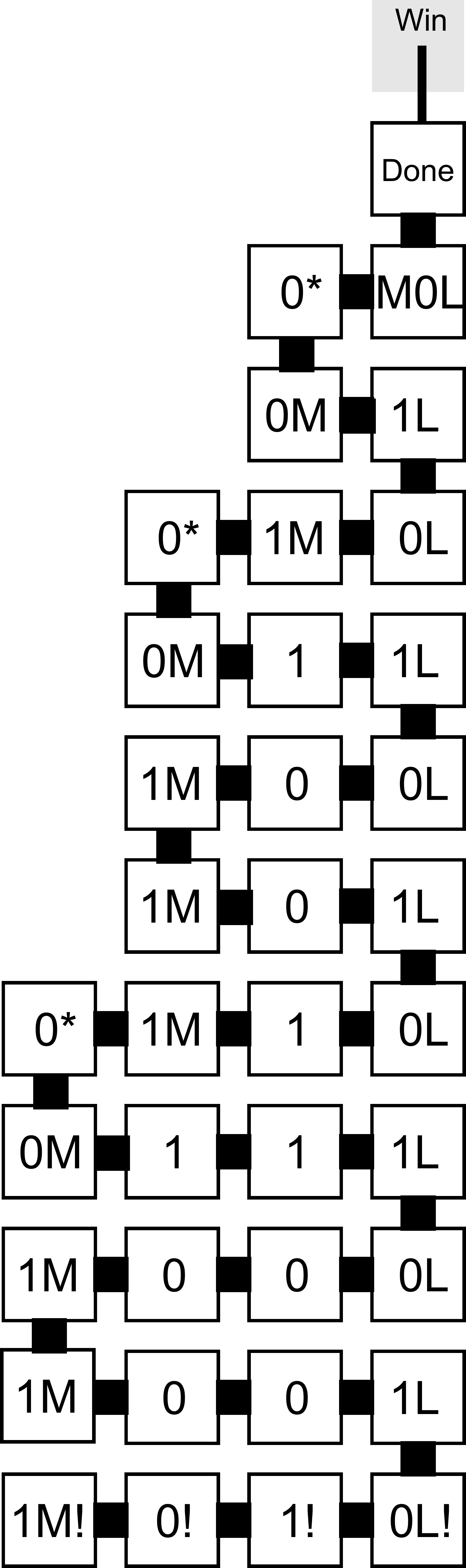}\hspace{4ex}}
    \quad\quad\quad
    \subfloat[After a probe wins, it grows a path back down, along its counterclockwise-most (here, rightmost) side using the information on its east to find the bottom position.  It then grows up into layer $L_1$ and over to the original row of the input superside and reads the information from that side (which was output to both $L_0$ and $L_1$) to begin forming the output layer.]{
        \label{fig:3D-sim-2D-probe-return}
        \hspace{1ex}\includegraphics[height=5.0in]{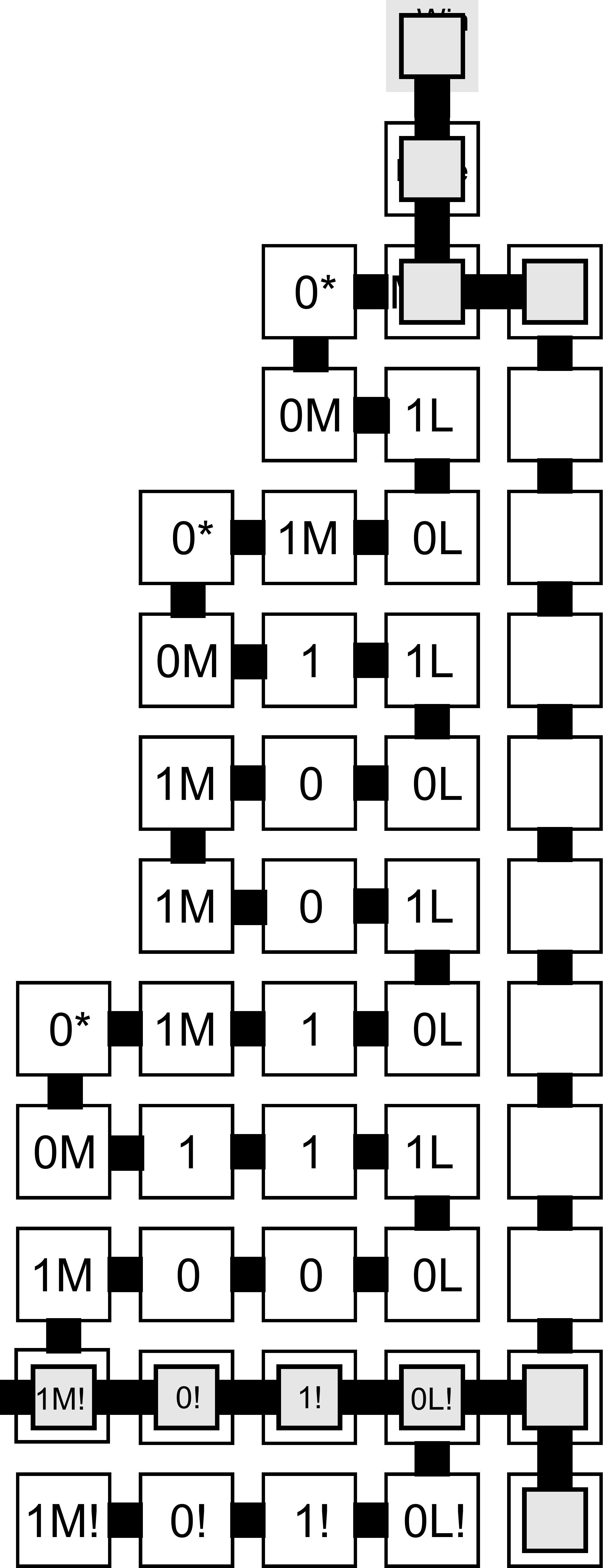}\hspace{1ex} }
    \caption{\small Example of a probe which begins with $h = 10$.  Each white and grey square represents a read-write gadget, and all of which output to the north of $L_0$, while those in the least significant bit position (with labels including `$L$') also output to the east of $L_0$.}
\label{fig:3D-sim-2D-probe}
\end{figure}

\subsubsection{The information layer}

The main purpose of the information layer, $L_1$, is to facilitate the transfer of information between the two other layers of the construction. Figures~\ref{fig:3D-sim-2D-output-level1}-\ref{fig:3D-sim-2D-output-level4} show in yellow the portions of a supertile that grow in into this layer.  Other than essentially ``bridging the gap'' between layers $L_0$ and $L_2$, growth in $L_1$ is required at the tip of the victorious probe in order to climb over the winning probe in case there are other, losing, probes surrounding the tile placed in the center position of the supertile.

\subsubsection{The output layer}

After a superside has won the competition via its probe, a path grows back down along the probe until reaching the base, at which point it begins growth in the clockwise direction.
It grows in a zig-zag path which rotates the encodings of $T$, $h$, and the appropriate new value for the set $S$ on that side, into position to create an output superside (see Figure~\ref{fig:3D-sim-2D-output-creation} for an extremely high-level sketch of the process, and see Figure~\ref{fig:3D-sim-2D-zig-zig-rotate-split} for a basic example of how the zig-zag pattern of growth of read-write gadgets can be used to perform operations such as rotation and splitting of information into two directions). The first step is to select a tile type (by its number) from the set $S$, which is represented in the input superside.  This is done by the first row to grow across the superside.

After reaching the beginning of $S$, at every position where a character~$y$ is encountered, the tile-selection row can nondeterministically choose to select the tile number $t$ immediately following the $y$.  If it chooses $t$, then the bits of $t$ are marked as selected and the selection is complete.  Otherwise, the same choice is possible for each $y$ encountered.  If (the number of) no tile has been selected when the $f$ symbol is encountered (there is guaranteed to be exactly one $f$, otherwise an input superside would not have been created), then this last tile type immediately following the $f$ marker is forced to be selected since it is the last valid choice and a choice must be made. Note that there could be multiple tiles to choose during this process. All entries marked with $y$, in the case that the system being simulated is nondeterministic, could be selected to attach at this step. Note that this type of nondeterministic selection of tile types does not fairly choose between all choices with equal probability (for the sake of discussion, assigning equal probability $1/k$ for each of $k$ nondeterministic tile choices for a given binding event).  This method is used for simplicity of discussion, but more complex selection methods, which choose options with closer to uniform probability, could be utilized.  The reader is encouraged to consult \cite{RNSSA} for a discussion of such ``random number selection'' techniques. These techniques can be implemented using zig-zag growth patterns of read-write gadgets.

Before selecting the bits of the tile type to be simulated, the supertile represents the empty space, i.e., a point in the simulated system that has yet to receive a tile type. However, once the tile-selection row has selected all the bits, we know what tile type it simulates. For a complete description of the representation function, see section \ref{sect:repr}. For the selected bits of $t$, since the encoding of $T$ gets rotated and continues to move upward, the encoding of $T$ is available to have the bits of $t$ pass through it.  The bits identifying the number for each tile type encoded in $T$ are marked if they match the bits of $t$, and after all of the bits of $t$ have passed through the uniquely matching tile type number is identified.  This provides subsequent rows of growth the ability to select the encoding of the appropriate side (for the about-to-be-formed superside) for tile type $t$, so that they can then be rotated into position to become the set $S$ for the output superside.

Since only one input superside can possibly win the competition, and all growth initiated from a side which lost the competition remains in the competition and information layers ($L_0$ and $L_1$), it is guaranteed that the output layer is completely available for use by the winning superside to grow clockwise around the supertile and create the necessary output supersides.  It is important to note that one zig-zagging, one-tile-wide path of tiles is responsible for the growth originating from an input superside, growing the probe, claiming the center position of the supertile, growing back down to the input superside, selecting which tile to represent, moving and rotating the information for new output supersides around the supertile. We must use a single path in order to ensure that any information, which is implicitly represented by the geometry of the read-write gadgets, is in place before needing to be read said information.

At the point when the information necessary to form a new output superside is fully rotated and in position, which may then grow into an input superside for an adjacent supertile, then the single path splits into two paths.  The original path continues to transfer the information around the supertile for all other output supersides (terminating after placing the information for the third output superside), while the new branch is free to potentially begin the growth of and win the competition for the new supertile. However, the new path first grows along the gadgets representing the information for the new superside and checks the values of the new set $S$.  If the location for every tile in $S$ is marked with an $n$, then there is no tile in $T$ that can attach to this side of the tile being represented by the current supertile and the path building this output superside terminates before it starts to build a corresponding input superside. Otherwise, in the case where there is a tile that could attach to the newly formed output side, the new path continues by growing another row which copies the information for the output superside down to level $L_0$. If it is able to complete the growth of the new output superside, then an input superside in the region for an adjacent supertile assembles and begins the growth of the probe for that new supertile. However, if a supertile already exists in that neighboring position and had already placed an input superside into this side of the current supertile (which must have lost the competition for this supertile, else it would be the one creating the output supersides), then the path creating the new output superside will be blocked and will terminate. This is because the slight overlap of the regions representing $h$ (the positions for the least significant bits of each copy of $h$ are in the same location, which puts them in the correct alignment to grow probes directly toward the center of each supertile; see the west side of Figure~\ref{fig:3D-sim-2D-output-level2} for a depiction of how the locations for the two encodings of $h$, north (input) and south (output), overlap).
This correctly models simulation since a supertile must already exist in the adjacent position to have placed an input superside here, and therefore it is unnecessary to attempt to grow into that location.  Furthermore, if a supertile already exists in the adjacent location but has yet to place an input superside that will prevent the growth of this new superside, that will cause no problem either because the completion of the new superside will only result in a probe which grows toward the center of the adjacent supertile but fails to win the competition.  The resulting assembly will not break the simulation of the adjacent supertile (of course, if no probe has yet claimed the center position to win the competition, this superside has a valid chance at doing so).

As the information in the output layer grows clockwise around the supertile, the spacing is designed so that each completed rotation of the information for a superside provides an implicit ``counter'' that provides the output layer with the information necessary to know when to stop and deposit an output side.  Thus, no other counter values need to be encoded and the rotations of the information can provide all of the necessary spacing information for correct growth.

See Figures~\ref{fig:3D-sim-2D-output-level1}-\ref{fig:3D-sim-2D-output-level4} for an example of how an output layer grows.  The grey regions represent tiles in the competition layer, $L_0$, the yellow those in the information layer, $L_1$, and the blue those in the output layer, $L_2$.  In this series of figures, a scenario is shown where there are four input supersides, all vying for the center position of the supertile, with the southern probe winning that competition.  For the sake of depicting the full flow of information and location of information in all input and output supersides, the overlapping positions of encodings of $h$, which would prevent output supersides from forming where completed input supersides already exist, are ignored.  However, the fact that encodings from $h$'s of input supersides use space needed by the encodings of $h$ for the output supersides would actually prevent them from completing since they are unnecessary.

\begin{figure}
\begin{center}
\includegraphics[width=2.5in]{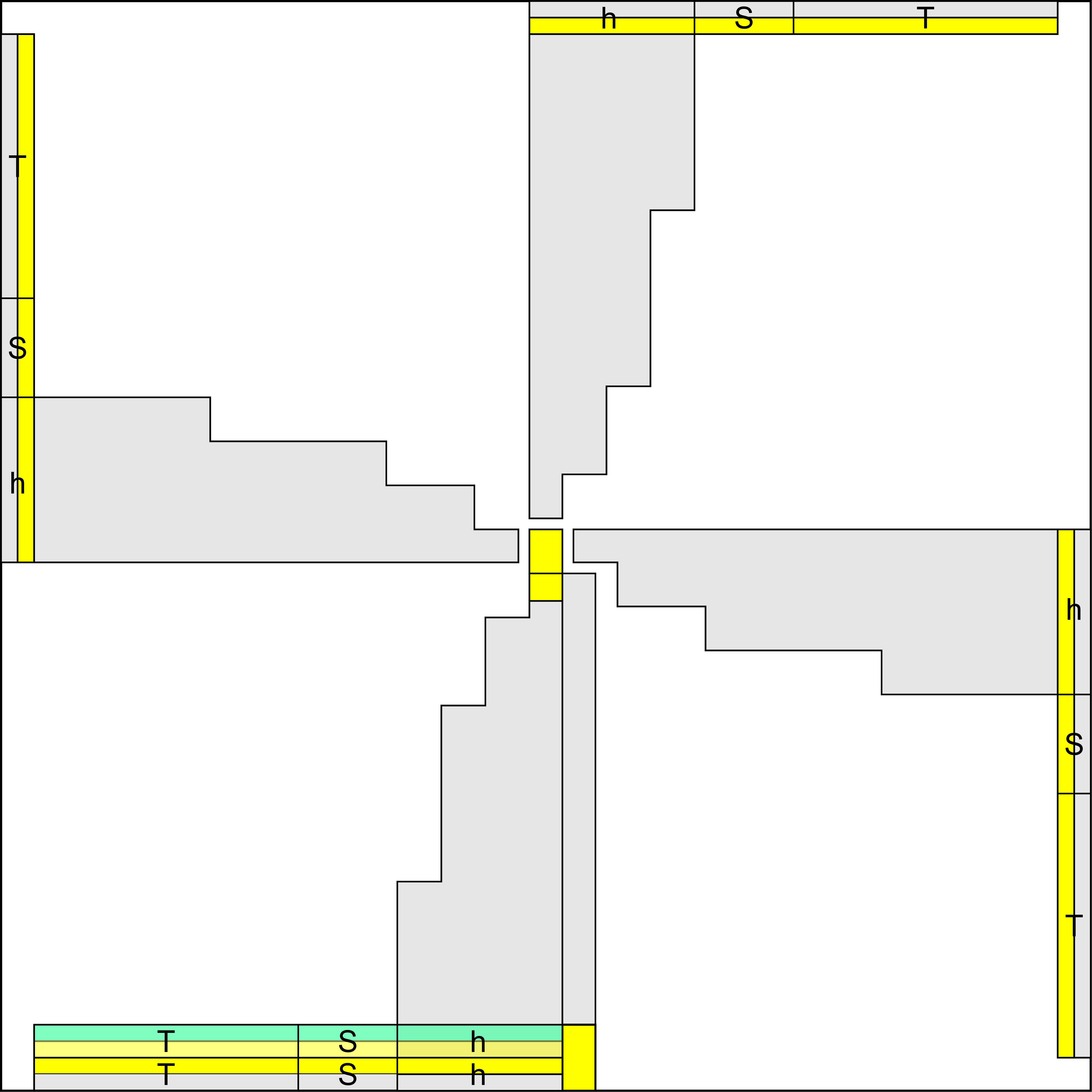}
\caption{A high-level depiction of the output layer (not to scale), part 1/4.  Grey regions represent portions in layer $L_0$, yellow those in $L_1$, and blue in~$L_2$.}
\label{fig:3D-sim-2D-output-level1}
\end{center}
\end{figure}

\begin{figure}
\begin{center}
\includegraphics[width=2.5in]{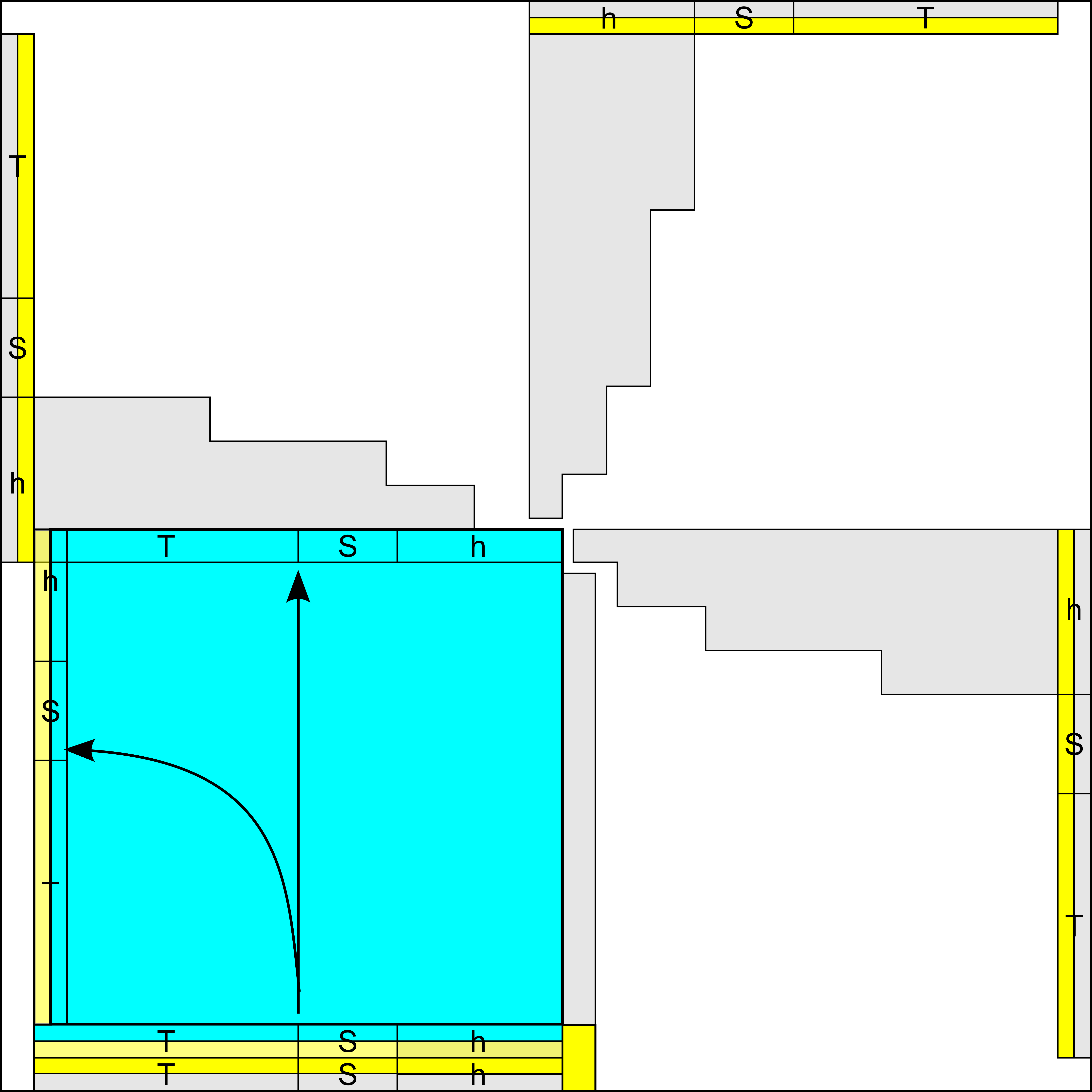}
\caption{A high-level depiction of the output layer shown in blue (not to scale), part 2/4.}
\label{fig:3D-sim-2D-output-level2}
\end{center}
\end{figure}

\begin{figure}
\begin{center}
\includegraphics[width=2.5in]{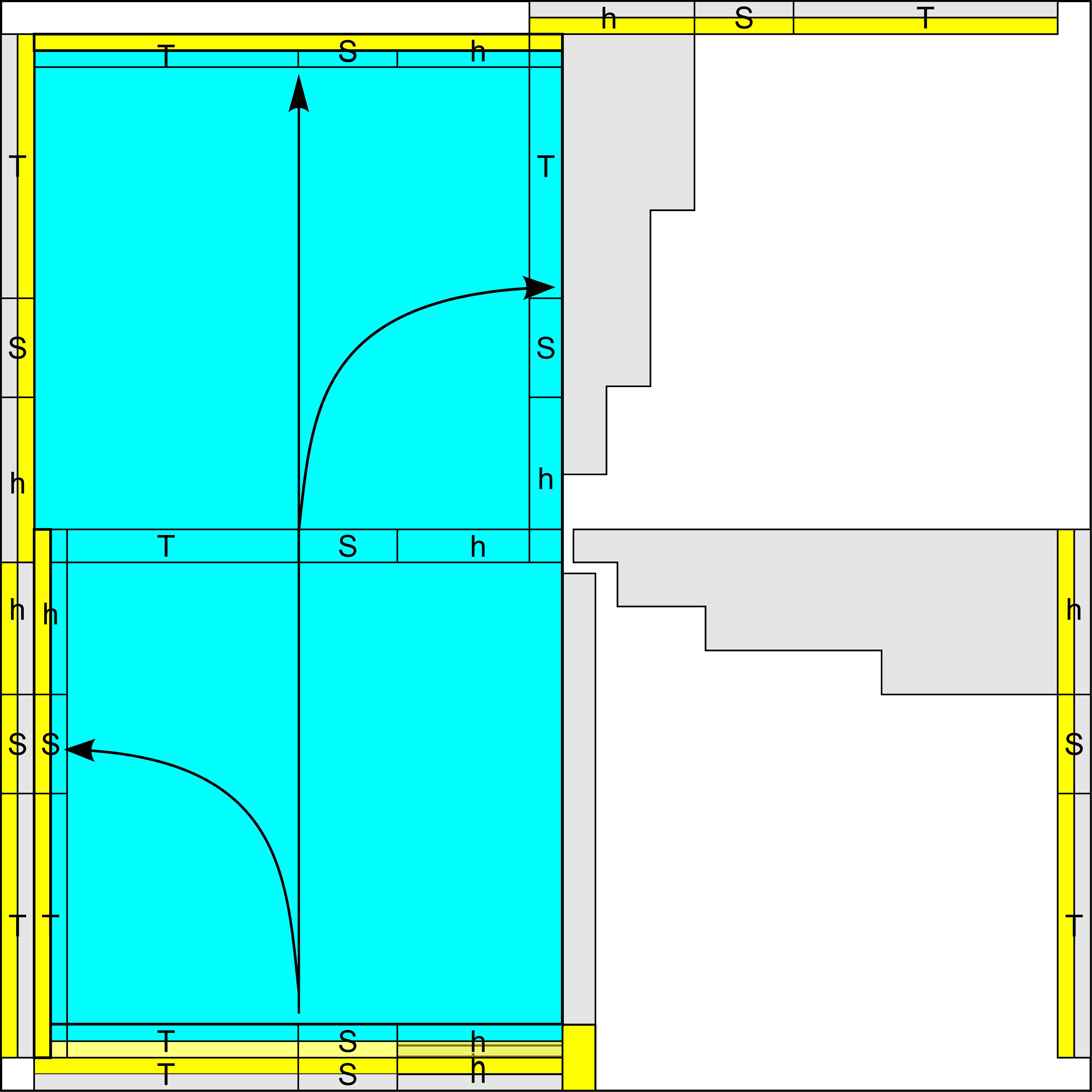}
\caption{A high-level depiction of the output layer (not to scale), part 3/4.}
\label{fig:3D-sim-2D-output-level3}
\end{center}
\end{figure}

\begin{figure}
\begin{center}
\includegraphics[width=2.5in]{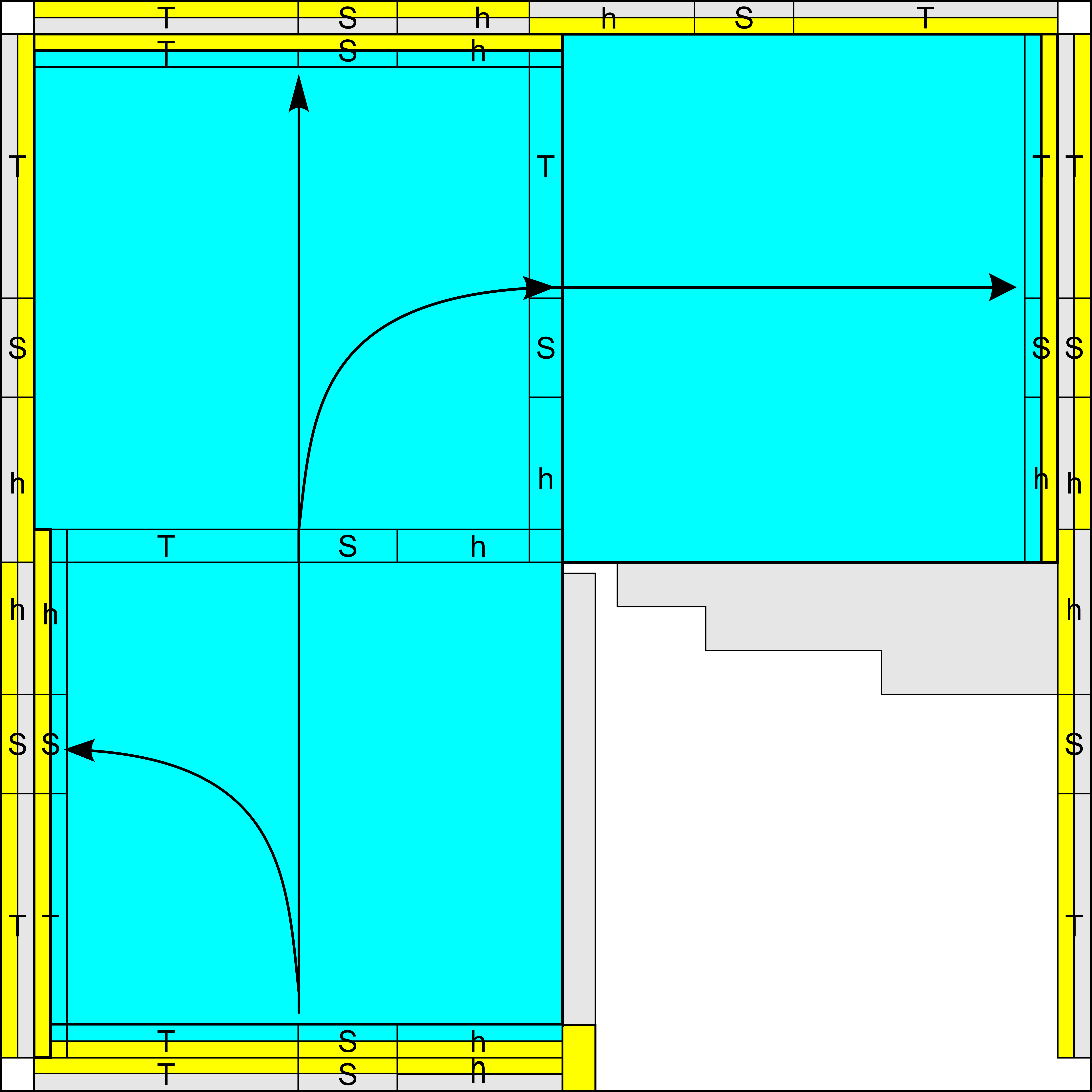}
\caption{A high-level depiction of the output layer (not to scale), part 4/4.}
\label{fig:3D-sim-2D-output-level4}
\end{center}
\end{figure}

\begin{figure}
\begin{center}
\includegraphics[width=3.0in]{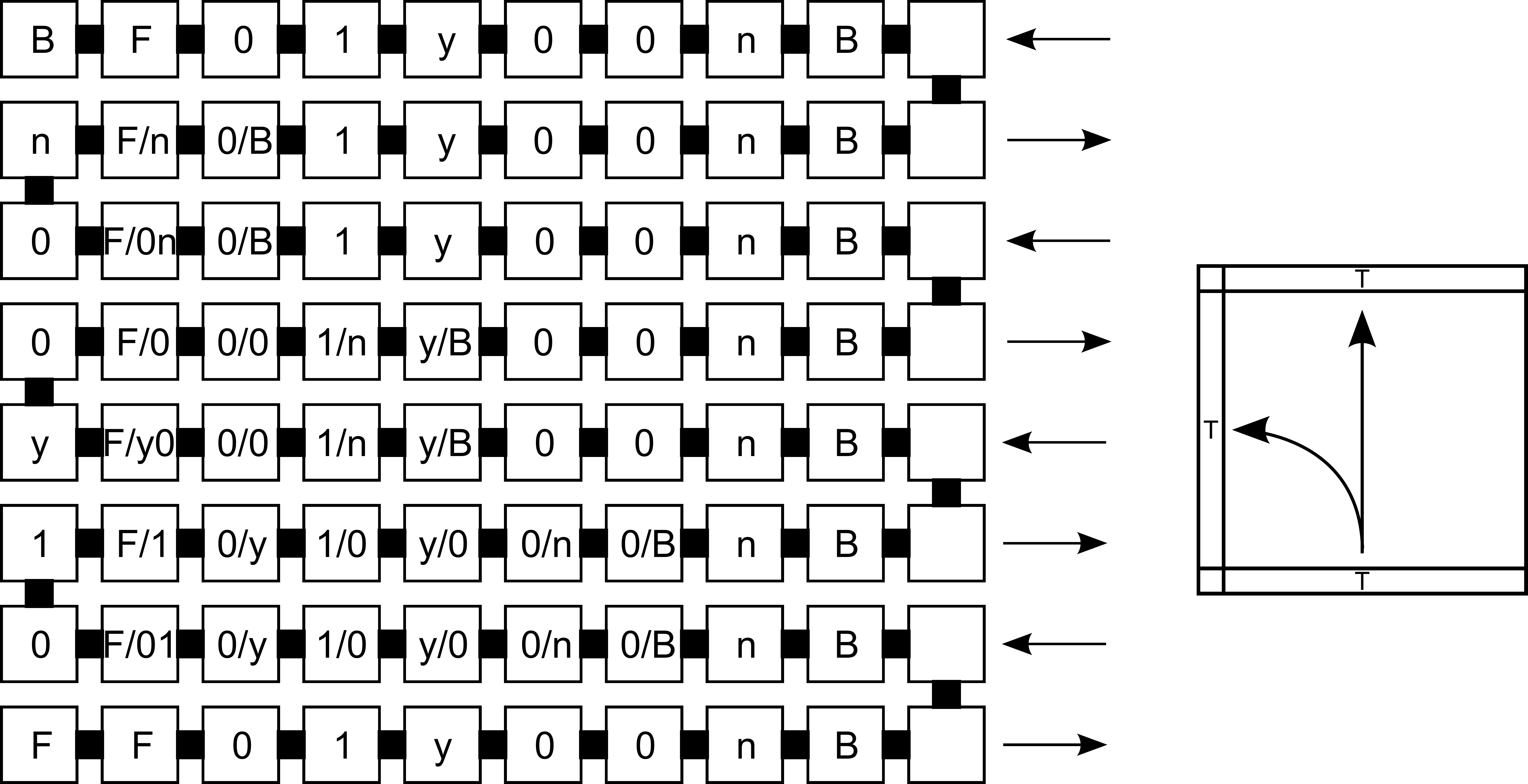}
\caption{A basic depiction of how information can be rotated and simultaneously carried forward by a pattern of zig-zag growth.}
\label{fig:3D-sim-2D-zig-zig-rotate-split}
\end{center}
\end{figure}

\begin{figure}
\begin{center}
\includegraphics[width=3.5in]{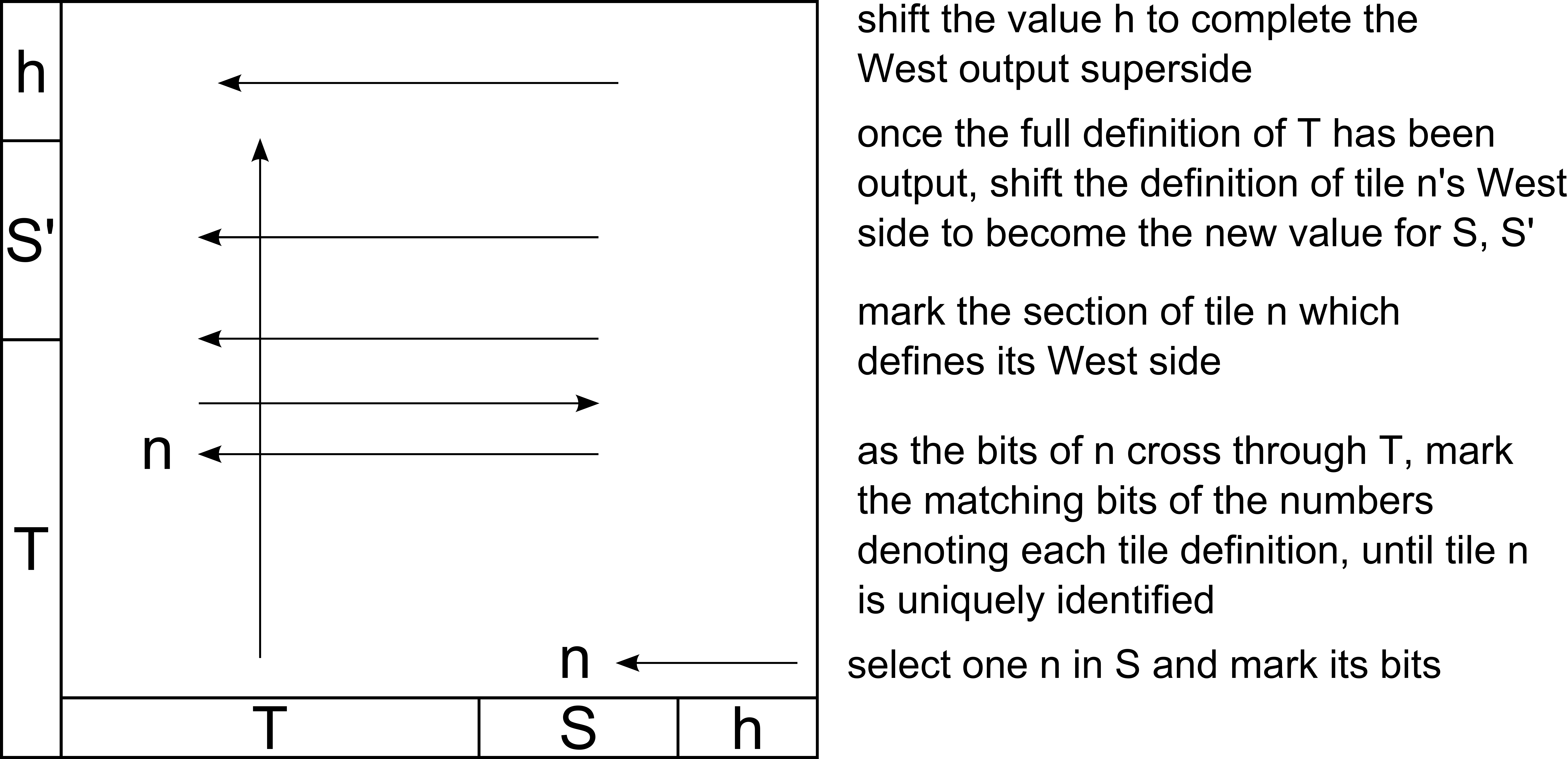}
\caption{A very high-level description of what occurs as the information from a superside propagates forward (upward in this figure) and is simultaneously rotated while the necessary information for the specific output side is selected and rotated into the correct position.}
\label{fig:3D-sim-2D-output-creation}
\end{center}
\end{figure}

\subsubsection{Seed structure}\label{sec:3D-sim-2D-seed-structure}

The seed structure $\sigma_T$ is a single supertile which maps to the seed tile $s \in T$.  It is the only supertile which has no input supersides.  Instead, it has one output superside corresponding to each side of $s$ which has a glue, with the structure of the output superside being identical to the structure of all other output supersides.  Specifically, the output superside consists of the outermost read-write gadgets along the perimeter of a given side, which would normally grow from an output layer.  In order to provide a connected seed structure, each output superside is connected to the center position of $\sigma_T$ by a single-tile-wide path of tiles, and the center position has a tile type unique to the central position of the seed tile.  See Figure~\ref{fig:3D-sim-2D-seed-structure} for an example.

\begin{figure}
\begin{center}
\includegraphics[width=2.0in]{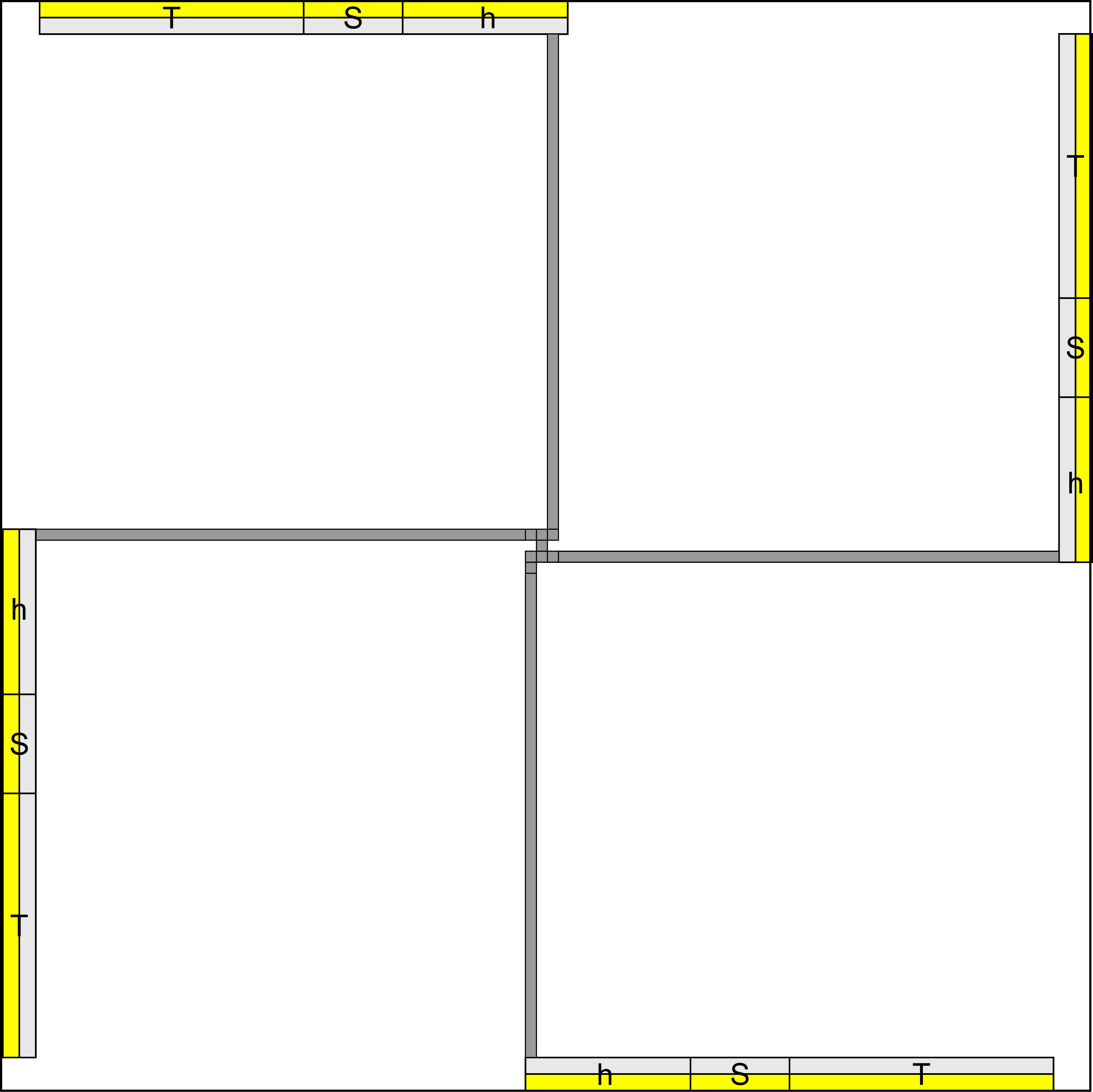}
\caption{The structure of the seed assembly $\sigma_T$.  Note that there would only be an output superside for each side of the seed tile which has a non-zero strength glue.  The other sides would have no tiles present in the seed structure.}
\label{fig:3D-sim-2D-seed-structure}
\end{center}
\end{figure}

\subsubsection{Scale factor of the simulation}

Here we present the scale factor for the simulation of $T$.  Note that when we discuss the amount of space required for the various encodings, each bit or character, rather than being represented by a single tile, is represented by a read-write gadget.  The read-write gadgets used by this construction are constant size, regardless of $T$, with widths and heights equivalent to the total number of different symbols used in the construction, and the total depth of the simulation is always $6$.  Since the increase in the scale caused by the use of read-write gadgets is constant, we ignore it for the rest of this discussion.

First, we discuss the length of the encoding of $T$.  For each $t \in T$, we encode the number assigned to $t$, which is of length $O(\log |T|)$.  Then, for each of the $4$ sides of $t$ we encode a single character ($N$,$E$,$S$,or $W$) and a list which contains the number assigned to each tile type and a single character ($y$, $f$, or $n$), for a length of $4 (1 + |T|(O(\log |T|)+1)) = O(|T| \log |T|)$.  Finally, for each tile there is a single character, $D$, at the end.  This results in a total encoding of size $O(|T| \log |T|)$ for each tile type.  Given the encoding of $|T|$ tile types and two more characters ($B$ and $F$), the full encoding of $T$ requires space $O(|T|^2 \log |T|)$.

The scale factor of the simulation is determined by the length of the sides of the supertiles.  Each side must be sized that that it can contain  1. a constant sized gap at each corner (size $O(1)$), 2. two copies of the encoding of~$T$ (size $O(|T|^2 \log |T|)$), 3. two encodings of $S$, each of which are the encoding of a single side of one tile type (size $O(|T| \log |T|)$), and 4. two copies of $h$, which is the $\log$ of the probe height. To determine the height to which a probe must grow, we first assume that the sides of each supertile contain only the copies of $T$ and $S$, which makes each side of width $2 O(|T|^2 \log |T|) + 2 O(|T| \log |T|) = O(|T|^2 \log |T|)$.  If all sides were of that length, to get to the center, a probe would need to grow to height $h' = O(|T|^2 \log |T|)/2 = O(|T|^2 \log |T|)$, which can be encoded in $\log h' = O(\log (|T|^2 \log |T|))$ space.  We then let $h = h' + (2 \log h')/2$ to account for the additional distance a probe must grow to account for the two copies of $h$ encoded in each side and note that $\log h = O(\log (|T|^2 \log |T|))$. (Note that to get side lengths and a value for $h$ which cause the probes to all grow to within exactly one distance of a read-write gadget from the exact center of the supertile, some ``padding'' of up to width $\log h$ may be added between the encodings of $S$ and $h$).  Thus the size of each side, and therefore the scale factor for the simulation, is $O(|T|^2 \log |T|)$.

\subsubsection{Representation function}\label{sect:repr}

The representation function $R$ for the simulation of a tile set $\mathcal{T}$ maps supertiles over $U$ to tiles of $T$ as follows.  For a supertile $s$, if there is no tile in the center location, it maps to an empty location.  If there is a tile in the center, if it is the special center tile for the seed, $s$ maps to the seed $\sigma$, otherwise $R$ follows the path back down the probe and to the point that a tile number is selected from the set $S$.  The tile number uniquely identifies the tile $t \in T$ that $s$ represents.

\section*{Acknowledgement}
We thank Robert Schweller for discussions on Theorem 3.1 of~\cite{Aggarwal-2005}.

\bibliographystyle{plain}
\bibliography{t1IU,tam}

\end{document}